%% file: main.tex
\documentclass{article}
\usepackage{graphicx} %

\input{head}

\title{$k$-Clustering via Iterative Randomized Rounding}

\newcommand{\authorspace}{5em}
\author{
  \hspace{\authorspace}
  Jaros{\l}aw Byrka\thanks{University of Wroc{\l}aw. \texttt{jby@cs.uni.wroc.pl}.}
  \and
  Yuhao Guo\thanks{Tsinghua University. \texttt{yh-guo23@mails.tsinghua.edu.cn}.}
  \and
  Yang Hu\thanks{Tsinghua University. \texttt{y-hu22@mails.tsinghua.edu.cn}.}
  \hspace{\authorspace}
  \and
  \hspace{\authorspace}
  Shi Li\thanks{Nanjing University. \texttt{shili@nju.edu.cn}.}
  \and
  Chengzhang Wan\thanks{Tsinghua University. \texttt{wcz23@mails.tsinghua.edu.cn}.}
  \and
  Zaixuan Wang\thanks{Nanjing University. \texttt{zaixuanwang@smail.nju.edu.cn}.}
  \hspace{\authorspace}
}
\date{}

\begin{document}

\maketitle

\input{abstract}

\input{intro}

\input{prelim}

\input{LMP}

\input{ProcessingLP}

\input{convertion}

\input{cost-proof-new}

\section*{Acknowledgment}
We would like to thank Aravind Srinivasan for valuable discussions in the early stage of this research project.

The work of Shi Li and Zaixuan Wang is supported by State Key Laboratory for Novel Software Technology, New Cornerstone Science Foundation, and Fundamental and Interdisciplinary Disciplines Breakthrough Plan of the Ministry of Education of China (No.JYB2025XDXM118). The work of Jaros{\l}aw Byrka is supported by Polish National Science Centre grant 2020/39/B/ST6/01641.

\bibliographystyle{plain} 
\bibliography{Iterative_UFL}

\input{pseudo-solution-to-solution}

\end{document}

%% file: head.tex
\usepackage[T1]{fontenc}
\usepackage[margin=1in]{geometry}
\usepackage{amsfonts,amsmath,amsthm,amssymb,mathtools}  %

\mathtoolsset{centercolon}
\usepackage{xfrac,nicefrac}
\usepackage{mathdots}
\usepackage{mleftright}  %
\let\left\mleft
\let\right\mright

\usepackage{algpseudocode}

\usepackage{xspace}
\xspaceaddexceptions{]\}}  %
\usepackage{regexpatch}

\usepackage{bm,bbm,dsfont}  %
\usepackage{caption}
\usepackage[normalem]{ulem}
\usepackage{enumitem}

\usepackage{graphicx}
\usepackage{float}
\usepackage{subcaption}  %
\usepackage{tcolorbox}
\usepackage{tikz}
\usetikzlibrary{decorations.pathreplacing}
\usetikzlibrary{calc}
\usetikzlibrary{positioning}
\usetikzlibrary{arrows.meta}
\usetikzlibrary{math}
\usetikzlibrary{patterns}

\usepackage[linesnumbered,boxed,ruled,vlined]{algorithm2e}

\SetCommentSty{mycommfont}

\newcommand{\R}{\mathbb{R}}

\allowdisplaybreaks

\usepackage{thmtools,thm-restate}
\usepackage[colorlinks,citecolor=blue,linkcolor=blue,urlcolor=red]{hyperref}

\theoremstyle{plain}

\newtheorem{theorem}{Theorem}[section]  %
\newtheorem{lemma}[theorem]{Lemma}

\newtheorem{corollary}[theorem]{Corollary}
\newtheorem{cor}[theorem]{Corollary}
\newtheorem{claim}[theorem]{Claim}

\theoremstyle{definition}  %
\newtheorem{definition}[theorem]{Definition}

\usepackage[capitalise]{cleveref}
\crefname{algocf}{Algorithm}{Algorithms}
\Crefname{algocf}{Algorithm}{Algorithms}
\crefname{claim}{Claim}{Claims}
\Crefname{claim}{Claim}{Claims}

\newfloat{Distribution}{htbp}{loa}
\crefname{Distribution}{Distribution}{Distributions}
\Crefname{Distribution}{Distribution}{Distributions}
\newfloat{Protocol}{htbp}{loa}
\crefname{Protocol}{Protocol}{Protocols}
\Crefname{Protocol}{Protocol}{Protocols}

\SetKwProg{Function}{Function}{:}{}

\SetKwProg{CodeBlock}{}{}{}
\SetKwProg{Repeat}{Repeat}{:}{}

\newcommand{\cI}{\mathcal{I}}

\DeclareMathOperator*{\E}{\mathbb{E}}
\newcommand{\Var}{\mathrm{Var}}
\newcommand{\Cov}{\mathrm{Cov}}

\usepackage{xcolor}         
\definecolor{AUXLblue}{RGB}{ 51,145,202}
\definecolor{AUXLlightblue}{RGB}{144,196,231} 

\definecolor{AUXLorange}{RGB}{246,174, 60}
\definecolor{AUXLlightorange}{RGB}{254,216,177}

\definecolor{AUXLpurple}{RGB}{135,119,175}
\definecolor{AUXLlightpurple}{RGB}{226,178,255}

\definecolor{AUXLyellow}{RGB}{255,221,  0}
\definecolor{AUXLlightyellow}{RGB}{255,238,127}

\definecolor{AUXLgreen}{RGB}{120,180,88}            
\definecolor{AUXLlightgreen}{RGB}{202,240,142} 

\definecolor{AUXLgray}{RGB}{105,106,107}             %
\definecolor{AUXLlightgray}{RGB}{208,209,211}  

\definecolor{AUXLred}{RGB}{219, 72, 72}              %
\definecolor{AUXLlightred}{RGB}{244,195,195}         %

\usepackage[colorinlistoftodos,prependcaption,textsize=tiny]{todonotes}
\usepackage{xargs} 

\newcommandx{\shi}[2][1=]{\todo[linecolor=AUXLgray,backgroundcolor=AUXLgray!25,bordercolor=AUXLgray,#1]{{\textbf{\underline{Shi}}:\\}~#2}}
\newcommandx{\chengzhang}[2][1=]{\todo[linecolor=AUXLpurple,backgroundcolor=AUXLpurple!25,bordercolor=AUXLpurple,#1]{{\textbf{\underline{Chengzhang}}:\\}~#2}}
\newcommandx{\yang}[2][1=]{\todo[linecolor=AUXLyellow,backgroundcolor=AUXLyellow!25,bordercolor=AUXLyellow,#1]{{\textbf{\underline{Yang}}:\\}~#2}}

\newcommandx{\jby}[1]
{{\color{green} #1}}

\newcommandx{\yuhao}[1]
{{\color{green} #1}}

\DeclarePairedDelimiter{\bk}{(}{)}

\newcommand{\av}{\mathrm{av}}
\newcommand{\ball}{\mathrm{ball}}

\newcommand\poly{\mathrm{poly}}

\newcommand{\cost}{\mathrm{cost}}

\newcommand{\imb}{\mathrm{imb}}

\newcommand{\force}{{\mathrm{force}}}

%% file: abstract.tex
\begin{abstract}
    In $k$-clustering problems we aim to partition points from the given metric space into $k$ clusters while minimizing a certain distance based objective function. We focus on centroid based functions that encode connection costs of points in a cluster to a facility being the center of the cluster. Popular functions include: the sum of distances to the center in the $k$-median setting, or the sum of squared distances to the center in the $k$-means setting.
    State-of-the-art approximation algorithms for these problems are obtained via sophisticated methods tuned for the specific cost function or metric space.

    In this work we propose a single rounding algorithm for the fractional solutions of the standard LP relaxation for $k$-clustering. 
   As a starting point, we obtain an iterative rounding $(\frac{3^p + 1}{2})$-Lagrangian Multiplier-Perserving (LMP) approximation for the $k$-clustering problem with the cost function being the $p$-th power of the distance. Such an algorithm outputs a random solution that opens $k$ facilities \emph{in expectation}, whose cost in expectation is at most $\frac{3^p + 1}{2}$ times the optimum cost. Thus, we recover the $2$-LMP approximation for $k$-median by Jain et al.~[JACM'03], which played a central role in deriving the current best $2$ approximation for $k$-median.  Unlike the result of Jain et al., our algorithm is based on LP rounding, and it can be easily adapted to the $L_p^p$-cost setting.  For the Euclidean $k$-means problem, the LMP factor we obtain is $\frac{11}{3}$, which is better than the $5$ approximation given by this framework for general metrics.

    Then, we show how to convert the LMP-approximation algorithms to a true-approximation, with only a $(1+\varepsilon)$ factor loss in the approximation ratio. 
    We obtain a ($\frac{3^p + 1}{2}+\varepsilon$)-approximation algorithm  for $k$-clustering with cost function being the $p$-th power of the distance, for $p \geq 1$. This reproduces the best known ($2+\varepsilon$)-approximation for $k$-median by Cohen-Addad et al. [STOC'25], and improves the approximation factor for metric $k$-means from 5.83 by Charikar at al. [FOCS'25] to $5+\varepsilon$ in our framework. Moreover, the same algorithm, but with a specialized analysis, attains ($4+\varepsilon$)-approximation for Euclidean $k$-means matching the recent result by Charikar et al. [STOC'26].

    Our algorithm not only is a single solution to multiple settings, it is also conceptually simpler than the previously best approximations. The main idea is to use a natural iterative randomized rounding procedure on a fractional solution to the standard LP-relaxation. A careful implementation of this procedure allows to produce solutions with $k+O(1)$ clusters. It remains to combine this result with the reduction by Li and Svensson [STOC'13] to obtain the desired result.
\end{abstract}

%% file: intro.tex
\section{Introduction}

In $k$-clustering we are typically given a finite set of demand points $C$ (sometimes called clients) and a set of possible locations of cluster centers $F$ (sometimes called facilities). The goal is to partition the demand points into $k$ clusters by selecting $k$ of the centers from $F$ and assigning each client in $C$ to one of the selected centers. The cost of a clustering is usually defined with respect to a distance-based cost function, hence it is commonly assumed that the demand points and centers come from a given metric space (i.e. distance $d(i,j)$ is defined for all $i,j \in C \cup F$).
In this work we will focus on the setting where the cost of the clustering is expressed as the sum of the assignment costs of clients to their cluster centers and the cost of a single client $j$ assigned to facility $i$ is computed as the $p$-th power of the $d(i,j)$ distance.

Clustering has been the topic of research in various computational contexts, ranging from Operations Research to Statistics~\cite{macqueen1967multivariate}. Particularly popular recently is the application of $k$-clustering algorithms in the context of classifying high dimensional data (see, e.g.,~\cite{jain2010data}). In the context of data classification the squared distance function and the Euclidean metric setting (called Euclidean $k$-means, or sometimes simply $k$-means) appear to be particularly relevant.

Computing optimal clustering (in most of the relevant settings) is generally NP-hard. In this work we focus on approximation algorithms that, in polynomial time, compute $\lambda$-approximate solutions for some constant parameter $\lambda$. Most of the existing literature on approximation algorithms for clustering problems is focused on a specific cost function and metric space.
We will first discuss the know results for the most relevant settings. Then we will discuss our contribution in Section~\ref{sec:our_results}, which will be followed by a brief survey of the other closely related work in Section~\ref{sec:further_rel_work}. Then, we will give an informal high-level description of our approach in Section~\ref{sec:overview_of_techniques}.

\subsection{Previous results} \label{sec:prev_res}

Many of the algorithms for $k$-clustering discussed below utilize the concept of Lagrangian relaxation replacing the hard constraint of creating at most $k$ clusters with a cost associated with the number of created clusters. By an LMP $\lambda$-approximation algorithm we mean an algorithm for the Lagrangian relaxation that has factor 1 for the cost associated to the number of clusters and factor $\lambda$ for the service cost. Intuitively, such an LMP algorithm can be used to produce randomized clustering with expected number of clusters equal $k$, and in some cases the randomized solution is just a combination of two deterministic solutions referred to as a \emph{bi-point solution}.

\paragraph{$k$-median}
  The first constant factor approximation algorithm for $k$-median was obtained by Charikar et al.~\cite{charikar1999constant}, the factor was $6\frac{2}{3}$ and it was obtained via LP-rounding. Then Jain and Vazirani~\cite{jain2001approximation} gave a 6 approximation via an LMP 3-approximation primal-dual algorithm and a factor 2 bi-point solution rounding. Jain et al.~\cite{jain2003greedy} obtained a greedy factor 2 LMP algorithm, which combined with bi-point rounding resulted in 4-approximation for $k$-median. Arya et al.~\cite{arya2001local} gave a ($3+\varepsilon$)-approximation algorithm based on local search.

  Further progress was made by improving the bi-point rounding algorithms. Li and Svensson~\cite{li2013approximating} proposed a factor 1.366 bi-point rounding algorithm that opens $k+O(1)$ facilities and also introduced a reduction allowing to use such a pseudoaproximation algorithm to obtain a true approximation algorithm opening $k$ facilities. We will utilize such reduction also in the this paper. The construction allowed to get 2.733-approximation for $k$-median. Further improvements in bi-point rounding opening $k+O(1)$ facilities were obtained in~\cite{byrka2017improved} and~\cite{gowda2023improved}.

  Recently, Cohen-Addad et al.~\cite{cohen20252+} managed to avoid the loss in the approximation ratio from bi-point rounding. They managed to adapt the JMS algorithm to prevent opening too many additional facilities while maintaining factor $2+\varepsilon$ approximation. The number of extra facilities in their approach is super-constant, which prevented simply using the reduction from~\cite{li2013approximating}. Instead, they developed a different algorithm for instances whose cost is highly sensitive to the number of open facilities.

\paragraph{$k$-means}
    Constant factor approximation algorithm for $k$-means can be obtained via local search. Gupta and Tangwongsan~\cite{gupta2008simpler} studied such algorithms and obtained ($25+\varepsilon$)-approximation for general metric and ($9+\varepsilon$)-approximation for the Euclidean metric. With a highly non-trivial extension of the primal-dual algorithm of Jain and Vazirani~\cite{jain2001approximation}, Ahmadian et al.~\cite{ahmadian2019better} obtained ($9+\varepsilon$)-approximation in general metric and 6.357-approximation in the Euclidean metric. They also discussed that, following~\cite{feldman2007ptas}, a $\rho$-approximation algorithm for the discrete variant (where facilities are selected from a finite set) translates into a $(\rho + \varepsilon)$-approximation algorithm for the continuous variant (where facilities are selected from the entire continuous space). The factor for the Euclidean case was later improved by Cohen-Addad~\cite{cohen2022improved} to 5.92. Squared metric version of $k$-clustering was challenging partly because the greedy JMS algorithm~\cite{jain2003greedy} is not an LMP algorithm for $k$-means.

    Very recently new results for both settings were obtained by Charikar at al. First, factor 5.83-approximation was obtained for general metrics~\cite{charikar2025improved}. Next, ($4+\varepsilon$)-approximation was obtained~\cite{charikar2026} for Euclidean $k$-means.
    Both these results combine an adaptation of the JMS algorithm with an algorithm for instances whose cost is highly sensitive to the number of open facilities, just like in~\cite{cohen20252+}.

\paragraph{Hardness of approximation.} Jain et al.\cite{jain2003greedy} showed that it is NP-hard to approximate $k$-clustering with cost being $p$-th power of the distance with a factor better than $1+\frac{3^p-1}{e}$.
In particular, the lower bound for $k$-median is $(1+2/e) \approx 1.735$, and the lower bound for $k$-means is $1+8/e \approx 3.943$. Notably, the lower bound also holds for LMP approximation algorithms to the relaxed problem.

\subsection{Our results} \label{sec:our_results}

We propose a new iterative randomized rounding algorithm for $k$-clustering and obtain

\begin{theorem}
    For any $p\geq 1$, any sufficiently small constant $\varepsilon > 0$ that depends on $p$, there exists an $n^{(1/\varepsilon)^{O(p^2)}}$-time $(1+\frac{3^p-1}{2} + \epsilon)$-approximation algorithm for $k$-clustering with assignment cost being the $p$-th power of the distance. 
\end{theorem}

It immediately gives us
\begin{cor} \label{cor:k-median}
There exists a $(2+\varepsilon)$-approximation algorithm for $k$-median,    
\end{cor}
and
\begin{cor} \label{cor:k-means}
There exists a $(5+\varepsilon)$-approximation algorithm for (general metric) $k$-means.    
\end{cor}
Note that Corollary~\ref{cor:k-median} is an alternative way to obtain the main result from~\cite{cohen20252+}
and Corollary~\ref{cor:k-means} improves on the 5.83-approximation from~\cite{charikar2025improved}.

Under the assumption that the underlying metric space is Euclidean%
, we may further improve the analysis of our rounding procedure and obtain

\begin{theorem} \label{th:eulerian_k-means}
    For all $\varepsilon >0$ there exists a $(4+\varepsilon)$-approximation algorithm for $k$-means in the Euclidean metric
\end{theorem}
and
\begin{theorem} \label{th:eulerian_LMP_k-means}
    There exists a $(\frac{11}{3})$-LMP approximation algorithm for %
    $k$-means in the Euclidean metric, where candidate facilities are given explicitly. 
\end{theorem}

The result in Theorem~\ref{th:eulerian_k-means} reproduces the ($4+\varepsilon$)-approximation from~\cite{charikar2026}. The LMP approximation from Theorem~\ref{th:eulerian_LMP_k-means} for the Euclidean metric $(\frac{11}{3}) \approx 3.67$ is strictly better than the lower bound on the approximation ratio for general metric $1+8/e \approx 3.943$ from~\cite{jain2003greedy}.

\subsection{Further related work} \label{sec:further_rel_work}

The literature on clustering is very broad and it is not possible to discuss all important results compactly. Below we briefly discuss a fragment related to approximation algorithms that appears most relevant.

\paragraph{Facility location.}
Closely related to $k$-clustering is facility location, where instead of a hard constraint to open exactly $k$ facilities (to form exactly $k$ clusters) we have a location-specific cost of opening a facility. Special case where opening a facility in each location costs the same can be seen as a Lagrangian relaxation of $k$-clustering.
Numerous variants of location problems were studied, here we will only discuss the most basic variant called Uncapacitated Facility Location (UFL). In UFL we are give a set of clients $C$, a set of facilities $F$, facility opening cost $f:F\rightarrow \mathcal{R}_{\geq 0}$, and a distance function $d$ on $ C\cup F$ ; we must select a subset of facilities to open and assign clients to opened facilities. The goal is to minimize the total cost of facility opening and client connection distances. UFL and $k$-median problems are closely related.

The first constant factor approximation algorithm for UFL was the LP-rounding by Shmoys et al.~\cite{shmoys1997approximation}. Important contributions include: greedy algorithms~\cite{guha1999greedy,jain2003greedy}, local-search~\cite{arya2001local}, primal-dual~\cite{jain2001approximation}, improved LP-rounding~\cite{chudak2003improved, byrka2010optimal}. The best know approximation ratio of 1.488 was obtained by Li~\cite{li20131} and it is very close to the lower bound of 1.463 shown by Guha and Khuler~\cite{guha1999greedy}. Approximating UFL appears slightly easier than approximating $k$-median, since in UFL an approximation algorithm may open more facilities than in the optimal solution.

\paragraph{Other clustering objective functions.}
In this work we focus on cost functions that sum the distances from each client to the center of its cluster. Other types of cost functions were also considered in the context of clustering problems. Examples include, min-sum-radii clustering~\cite{behsaz2015minimum} and Max-k-Diameter clustering~\cite{fleischmann2025inapproximability}. It is unclear at this point whether our iterative randomized rounding approach may lead to valuable algorithms in the context of such cost functions.

\paragraph{Approximation in FPT time.}
While the focus of this paper is on algorithms that are polynomial in the size of the instance, regardless of the requested number of clusters $k$, there are also works that allow the running time of the algorithm to be $\poly(n)\cdot f(k)$ for some function $f$. Such algorithms can be valuable for instances of the problem with smaller values of $k$ and are often referred to as running in time FPT(k).
Notably, the lower bound of $1+\frac{3^p-1}{e}$ from Jain et al.\cite{jain2003greedy} on approximability holds also for FPT(k) time algorithms, and the lower bound can be matched with an FPT time approximation algorithm, see~\cite{cohen2019tight}.

\subsection{Overview of the techniques} \label{sec:overview_of_techniques}

\paragraph{Iterative rounding for LMP approximation.} We give the main idea behind our iterative randomized rounding algorithm for the Lagrangian Multiplier Preserving (LMP) algorithm for $k$-clustering. The rounding procedure depends only on the facility opening vector $y$ and the metric restricted to the set $F$ of facilities. 

Given the fractional opening $y$, we construct a directed \textit{neighborhood graph} over $F$. In this graph, each facility $i$ receives incoming edges from its nearest 1 fractional facility. For simplicity, we assume no facility splitting is needed in this construction. Then every facility $i$ has a fractional in-degree of $1$, where an edge $i’i$ contributes a degree of $y_{i'}$.

Consider an iteration where a single facility $i$ is selected to be opened with probability proportional to its $y_i$-value. Upon opening $i$, we permanently close all of its out-neighbors, by changing their $y$-values to $0$. Because the weighted average out-degree over all facilities is also $1$, this operation maintains the expected number of open facilities. By repeating this procedure until $y$ becomes integral, we obtain a solution with $k$ open facilities in expectation.

While less trivial, a compact argument shows that the expected connection cost from a client to the nearest open facility is at most $\frac{3^p + 1}{2}$ times its fractional cost. This is proved using an inductive potential function argument, where we show that the expected value of the potential function does not increase in a single iteration.

\paragraph{Converting the algorithm into a true approximation.} As in \cite{charikar2025improved,cohen20252+,li2013approximating}, the main obstacle in obtaining a true approximation for $k$-clustering is to guarantee that our algorithm always opens $k+O_{\varepsilon, p}(1)$ facilities, not just in expectation, with only a $1+\varepsilon$ loss in the approximation ratio.  We slightly relax the requirement, so that the event occurs with probability at least $1 - O(\varepsilon)$. This is sufficient for our purpose. 

For our rounding algorithm to work, we need the fractional solution $(x, y)$ to be $\varepsilon^c$-integral for some constant $c$. For now, we assume the property holds and show how it can be achieved later.  The main idea behind our modified iterative algorithm is that each iteration selects a set $I$ of facilities, rather than a single facility, such that each facility is in $I$ with a large probability proportional to $y_i$. As in the LMP algorithm, we open facilities in $I$ and close the out-neighbors of $I$. As the probabilities of including facilities in $I$ are large, we can terminate the algorithm in $O_{\varepsilon, p}(1)$ iterations. 

The crucial requirement is to guarantee that, with large probability in each iteration, the fractional number of open facilities does not increase by too much. We partition the facilities into three sets: $F^+, F^0$ and $F^-$, containing facilities with positive, 0 and negative imbalances respectively, where the imbalance of a facility $i$ is $1$ minus its fractional out-degree in the neighborhood graph. This is the net change in the number of open facilities in the LMP algorithm if $i$ is selected. We choose facilities from sets $F^+ \cup F^-$ and $F^0$ using two different procedures: unbalanced-update and balanced-update. In an iteration, we call one of the two procedures, each with probability $1/2$. 

The balanced-update procedure chooses facilities in $F^0$. For a facility $i \in F^0$, choosing $i$ in the LMP algorithm will not change the number of open facilities. However, choosing a set $I$ of facilities in $F^0$ in our modified algorithm presents a challenge. If the out-neighbors of $I$ overlap, then we remove fewer facilities than required, causing an increase in the number of open facilities. To resolve this, we guarantee that $I$ is an independent set, meaning that they have disjoint sets of out-neighbors. We construct $I$ using a randomized greedy procedure, which guarantees that the probability any $i \in F^0$ included in $I$ is large and approximately proportional to its $y_i$ value. 

In the unbalanced-update procedure, we choose a set $I$ from $F^+ \cup F^-$. Unlike balanced-update, these facilities can be added to $I$ independently. We give slightly larger probabilities for facilities in $F^-$ than those in $F^+$. This addresses the issue caused by overlapping out-neighborhoods.  Moreover, if the absolute imbalance of each facility $i \in F^-$ is not too big, concentration bounds allow us to prove that the net increase in opening facilities in the iteration is small with high probability. Fortunately, there are only a few facilities with big negative imbalances. We force them to open and then partition their sets of out-neighbors into smaller subsets using ``fictitious'' facilities. This ensures that concentration bounds can be applied.

The $\varepsilon^c$-integrality property guarantees that the fractional values for each facility and edge are bounded away from zero. This is critical for three reasons. First, if two facilities in $F^0$ have an overlap between their out-neighbors, then the overlap is ``large'' relative to the whole sets. Then the conflict graph over $F^0$ has a small degree, which allows us to choose a large enough independent set $I$. Second, for a facility $i \in F^+ \cup F^-$, its absolute imbalance is not too small. By assigning slightly larger probabilities for $F^+$, we gain a sufficiently large benefit that can cover the loss caused by the overlapping out-neighbors and ensure the application of concentration bounds. Third, we force facilities with large negative imbalance to open during our rounding algorithm. As the facilities have $y$-values at least $\varepsilon^c$, the total number of forcibly open facilities can be bounded by $O_{\varepsilon, p}(1)$. 

We run the algorithm only for a fixed number (which is $O_{\varepsilon, p}(1)$) of iterations. The total number of facilities forced to open is $O_{\varepsilon, p}(1)$. Furthermore, concentration bounds ensure that the number of normally open facilities is at most $k + O_{\varepsilon, p}(1)$ with $1-O(\varepsilon)$ probability. The iterative procedure does not guarantee an integral solution in the end. We treat the remaining fractional part as a fractional solution to a weighted $k$-center instance, and round it using a $3$-approximation. As a client will be connected in the iterative procedure with high probability, and the fractional solution is $\varepsilon^c$-integral, the loss incurred by this final step is negligible. 

Finally, despite the modifications, the $\frac{3^p+1}{2}$-approximation for connection costs can still be proved using the potential function argument. However, the analysis need to be adjusted to handle the approximate selection probabilities of facilities included in $I$. Also, we need the property that the events two different facilities included in $I$ are approximately independent, which is guaranteed by the two procedures. 

The algorithm opens $k+O_{\varepsilon, p}(1)$ facilities and gives good expected connection cost for each client. It remains to combine such algorithm with the reduction (form pseudo-approximation to approximation) by Li and Svensson from~\cite{li2013approximating}. Although their reduction was originally tailored for $k$-median, it generalizes naturally to the $L_p^p$ objective for any $p \geq 1$.

\paragraph{Ensuring $\varepsilon^c$-integrality of fractional solution via pipage rounding.}
We now describe the process of making the fractional solution $\varepsilon^c$-integral. We apply a randomized pipage-rounding routine which approximately preserves the sum of opening on a family of disjoint balls, carefully selected using a filtering procedure. For most clients, concentration bounds ensure that the expected multiplicative loss is at most $1 + O(\varepsilon)$.  However, there are special clients $j$, for which a $1 - \poly(\varepsilon)$ fraction of its connection is near $j$, but the remaining $\poly(\varepsilon)$ fraction is much farther away. For such clients, the rounding procedure would lose a factor of $2$ in the connection distance, and thus a factor of $2^p$ in the cost. Fortunately, we manage to align this pipage-rounding step with our iterative rounding procedure, so that the overall approximation ratio for these special clients remains bounded by $\frac{3^p + 1}{2}$.

\paragraph{Results for Euclidean $k$-means.} Our framework is versatile; while it supports general metrics for any $p \geq 1$, it can also be tailored to accommodate special metric spaces. The approximation ratio $\alpha$ is determined by both the power $p$ and the specific properties of the metric space. For the Euclidean $k$-means problem, our framework gives a $\frac{11}{3}$-LMP approximation. However, due to the special clients mentioned above, we could only get a $(2^2 + \varepsilon = 4 + \varepsilon)$-approximation for the problem.

%% file: prelim.tex
\section{Preliminaries}

The $k$-clustering problem studied in this work is defined as follows: We are given a set of clients $C$ (also called the set of points to be clustered), a set of facilities $F$ (also called the set of potential cluster centers), a metric distance function $d:C \cup F \times C \cup F \rightarrow \mathcal{R}_+$, a positive integer $k$, and a parameter $p \geq 1$. Our task is to select a subset of at most $k$ facilities $F' \subseteq F$ and a mapping $\sigma: C \rightarrow F'$. The goal is to minimize the assignment cost $\sum_{j\in C} d^p(j, \sigma(j))$.

When $p=1$, the problem is known as the $k$-median problem. When $p=2$ the problem is known as the $k$-means problem. In the literature the $k$-means problem is often studied under the additional assumption that the distance function is Euclidean (i.e., there exists an embedding of the points $C \cup F$ into - possibly high-dimensional - Euclidean space). To avoid confusion, we will use the term Euclidean $k$-median when referring to the setting with the assumption that the metric is Euclidean, and use the term metric $k$-median referring to the setting without such assumption.

Clustering problems are sometimes also considered in continuous spaces allowing the cluster centers to be chosen from a continuous space. In this work we follow the standard approach %
to focus on the discrete setting where cluster centers are selected from a given discrete set. At least in the context of the $k$-means objective, the reduction from the work of Feldman et al.~\cite{feldman2007ptas} can be used to derive analogous approximation results for the continuous setting. 

By modeling facility opening with a vector $y$ and the assignment $\sigma$ with variables $x_{i,j}$ for $i \in F, j \in C $ we obtain the following natural LP relaxation of the $k$-clustering problem.

\begin{align*}
\min \quad & \textstyle  \sum_{i \in F} \sum_{j \in C} d^p(i,j) \, x_{i,j}   \\
\text{s.t.} \quad & \textstyle \sum_{i \in F} y_{i} \leq k \\
& \textstyle \sum_{i \in F} x_{i,j} = 1, && \forall j \in C,  \\
& x_{i,j} \le y_i, && \forall i \in F,\, j \in C,  \\
& x_{i,j} \ge 0,\, y_i \ge 0, && \forall i \in F,\, j \in C, 
\end{align*}

In this work we will study algorithms that take a fractional solution $(x,y)$ feasible for the above linear program and produce a
feasible integral solution $(\hat{x}, \hat{y})$
of not much larger cost. We say a randomized algorithm is an LMP $\lambda$-approximation algorithm for $k$-clustering if it produces solution with $\E[\sum \hat{y}] \leq k$ and $\E[\sum_{i \in F} \sum_{j \in C} d^p(i,j) \, \hat{x}_{i,j}] \leq \lambda \cdot \sum_{i \in F} \sum_{j \in C} d^p(i,j) \, x_{i,j}$. If instead of condition $\E[\sum \hat{y}] \leq k$ the algorithm always satisfies $\sum \hat{y} \leq k + O(1)$, we call it a pseudo-approximation algorithm.

Li and Svensson~\cite{li2013approximating} showed that for the $k$-median problem there exists a reduction allowing to use a factor $\lambda$ pseudo-approximation algorithm to obtain a $(\lambda + \varepsilon)$-approximation algorithm. In this work, we utilize a natural extension of this reduction to the $L_p^p$ metric setting, see Appendix~\ref{sec:pseudo_approximation_reduction}. Additionally, we observe that it is sufficient for the condition $\sum \hat{y} \leq k + O(1)$ to be satisfied with high probability. We will therefore focus our attention on obtaining an algorithm that opens $k + O(1)$ facilities w.h.p.

Let $(V, d)$ be a metric space that is clear from the context. For  a set $U \subseteq V$, a point $j \in V$, and a radius $r \geq 0$, we use $\ball^\circ_U(j, r) := \{j' \in U: d(j, j') < r\}$ and $\ball_U(j, r) := \{j' \in U: d(j, j') \leq r\}$ to respectively denote the sets of points in $U$ with distance less than $r$ and at most $r$ to $j$. Given a directed graph $G$ and a vertex $v$ in $G$, we use $\deg^+_G(v)$ and $\deg^-_G(v)$ to denote the out and in degrees of $v$. For a %
vector $g \in \R^{\mathbb{D}}$ and a set $S \subseteq \mathbb{D}$, we define $g(S):=\sum_{i \in S}g_i$ unless otherwise stated.

\paragraph{Organization.} The rest of the paper is organized as follows. In Section~\ref{sec:LMP}, we give our LMP-approximation algorithm for the $k$-clustering problem, which achieves $\frac{3^p+1}{2}$-approximation for the problem with $L_p^p$ objective in general metrics, and $\frac{11}3$-approximation for discrete Euclidean $k$-means. Sections~\ref{sec:preprocessing}-\ref{sec:cost-new} describes and analyzes the true approximation algorithm. Section~\ref{sec:preprocessing} shows how to preprocess the LP solution $(x, y)$ into an $\varepsilon^c$-integral solution $(x'', y'')$, Section~\ref{sec:rounding_alg} describes the modified iterative rounding algorithm that opens $k+O_\epsilon(1)$ facilities with large probability, and Section~\ref{sec:cost-new} analyzes the cost incurred by the algorithm.  The reduction form pseudo-approximation to true approximation by Li and Svensson \cite{li2013approximating} was only given for the $k$-median problem. We show how it can be generalized to the $L_p^p$ objective for any $p \geq 1$ in Appendix~\ref{sec:pseudo_approximation_reduction}.

%% file: LMP.tex
\section{Generic LMP algorithm for $k$-clustering} 
\label{sec:LMP}

In this section, we describe the LMP algorithm for the $k$-clustering problem under the $L_p^p$ objective for any $p \geq 1$. The metric space may be general or restricted. The main theorem we prove is the following:
\begin{theorem}
    \label{thm:LMP}
    Let $\alpha \geq 1$ be a constant such that the following holds for any input instance. For every client $j \in C$, every set $T \subseteq F$, and every vector $y \in \R_{\geq 0}^T$, we have
    \begin{align}
        \sum_{i\in T,i'\in T}y_iy_{i'}\cdot\max\{\alpha \cdot d^p(i,j),(d(i,j)+d(i,i'))^p\}\le (2\alpha-1)\cdot y(T)\cdot\sum_{i\in T}y_id^p(i,j). \label{equ:distance_sum_goal}
    \end{align}
    Then there is a polynomial time $\alpha$-LMP algorithm for the problem. 
\end{theorem}

In Section~\ref{subsec:LMP-general-metric}, we shall show that for the $L_p^p$ objective in general metrics, we can set $\alpha = \frac{1 + 3^p}{2}$.

\subsection{Iterative rounding algorithm}

In this section, we give our iterative rounding algorithm for Theorem~\ref{thm:LMP}.

\begin{definition}[Neighborhood Graph]
    Given a vector $y'\in [0, 1]^F$, the neighborhood graph for $y'$ is a directed graph $(F, E, w)$ with possibly self-loops, and positive edge-weights $w \in (0, 1]^E$, defined using the following procedure. For every $i \in F$ with $y'_i > 0$, let $(w_{i'i})_{i' \in F} \in [0, 1]^F$ be the vector with the minimum $\sum_{i' \in F}w_{i'i} d(i', i)$ satisfying $\sum_{i' \in F}w_{i'i} = 1$ and $w_{i'i} \in [0, y'_{i'}]$ for every $i' \in F$. For any $i \in F$ with $y'_i = 0$, we let $w_{i'i} = 0$ for every $i' \in F$. Then we define $E$ to be the support of $w$, and restrict the domain of $w$ to $E$.

\end{definition}

So, for every $i \in F$ with $y'_i > 0$, the vector $(w_{i'i})_{i' \in F}$ is the 1 fractional nearest facility to $i$.  Treating $w_{i', i}$ as the fraction of the edge $(i', i)$, every $i \in F$ with $y'_i > 0$ has fractional in-degree $1$.  Notice that we always have $w_{i, i} = y'_i$ when $y'_i > 0$, as $i$ is the closest facility to itself. 

The randomized rounding algorithm is given in Algorithm~\ref{alg:LMP}. A notable feature of the algorithm is that it is independent of the set $C$ of clients and thus the fractional connection vector $x$. It only depends on $y$ and the metric over $F$. 
\begin{algorithm}[h]
    \caption{Iterative Rounding}
    \label{alg:LMP}
        $y' \gets y$\; 
        \While{$y'$ is not integral}
        {
            create the neighborhood graph $G = (F, E, w)$ for $y'$\;
            choose a facility $i' \in F$ randomly, with probabilities proportional to $y'_{i'}$ values\;
            for every $i \in \delta_G^+(i')$ do: with probability $\frac{w_{i'i}}{y'_{i'}}$, let $y'_i \gets 0$\;
            let $y'_{i'} \gets 1$\;
        }
        \Return $\{i \in F: y'_i = 1\}$
\end{algorithm}

To get some intuition on an iteration of the while loop, let us assume $y'_i > 0$ for every $i \in F$ and $w_{i'i} = y'_{i'}$ for every edge $i'i$. Then, in every iteration, we randomly choose a facility $i' \in F$ based on $y'$ values. We remove all out-neighbors of $i'$ by changing their values to $0$, and integrally open $i$ by changing its $y'$ value to $1$. When $w_{i'i} < y'_{i'}$, we remove $i'$ with probability $\frac{w_{i'i}}{y'_{i'}}$. Notice that our rounding algorithm is completely independent of the clients.

We say an iteration of the while loop is useful, when the $y'$-vector is changed. Clearly, if an iteration is useful, then for at least one facility $i \in F$, the value of $y'_i$ changes from fractional to integral. An iteration is non-useful when the facility $i'$ we choose already has $y'_i  = 1$, and we fail to change the value of any $y'_i$ to $0$. Using conditional distributions, we can avoid all the non-useful iterations and thus the algorithm always finishes in finite time.  However, it is instructive to analyze the version of the algorithm with non-useful iterations.

\subsection{Analysis of the number of open facilities} %

The analysis of the expected number of open facilities is easy. Focus on an iteration of the algorithm, let $F'$ be the set of facilities with $y'_i > 0$. When we choose $i' \in F'$ in the iteration, the expected decrement to $|y'|_1$ in Step 5 is $\sum_{i \in \delta^+(i')}\frac{w_{i'i}}{y'_{i'}}\cdot y'_i$, and the increase in Step 6 is $1$ (notice that we must have changed $y'_{i'}$ to $0$ in Step 5). Therefore, conditioned on choosing $i'$ in the iteration, the expected net increase in $|y'|_1$ is $\sum_{i \in \delta_G^+(i')}\frac{w_{i'i}}{y'_{i'}}\cdot y'_i - 1$.  The expected net increase in $|y'|_1$ over all choices of $i'$ is 
\begin{align*}
    \frac{1}{|y'|_1} \cdot \sum_{i' \in F'}y'_{i'}\left(\sum_{i \in \delta_G^+(i')}\frac{w_{i'i}}{y'_{i'}}\cdot y'_i - 1\right) = \frac{1}{|y'|_1} \left(\sum_{i'i \in E} w_{i'i}y'_i - |y'|_1 \right) = 0.
\end{align*}
The last equality used that $\sum_{i' \in \delta_G^-(i)}w_{i'i} = 1$ for every $i \in F'$, and thus $\sum_{i'i \in E} w_{i'i}y'_i = \sum_{i \in F'}y'_i = |y'|_1$. As the probability that the algorithm runs for infinite number of iterations is $0$, the expected number of open facilities at the end of the algorithm is $\E[|y'|_1] \leq k$.

\subsection{Analysis of connection cost}

In this section, we fix a client $j$ and analyze its expected connection cost. By splitting facilities at the beginning, we can assume $x_{ij} \in \{0, y_i\}$ for every $i \in F$. Then, we define $F_j$ to be the set of facilities $i$ with $x_{ij} = y_i > 0$.  So we have $y(F_j) = 1$.  We define $d_i = d(i, j)$ for every $i \in F_j$.

Let $\Phi_j$ be the random variable denoting the cost at the end of the algorithm. The main theorem we prove is the following:
\begin{theorem}
    $\E[\Phi_j] \le \alpha \sum_{i \in F_j}y_i d_i^p$. \label{thm:cost-analysis}
\end{theorem}

When some facility $i \in F_j$ is chosen in an iteration, we directly let $d_i^p$ be the cost of $j$, without looking at the facilities that may be open in the future. We say $j$ is happily connected in this case. When $j$ is not happily connected, we define its state to be $(S, b)$, where $S \subseteq F_j$ is the set of alive facilities in $F_j$, and $b$ is the distance from $j$ to its nearest integrally open facility.  We say a facility $i \in F_j$ is alive if we have $y'_i > 0$. Notice that if $i \in F_j$ is alive and $j$ is not happily connected yet, then we have $y'_i = y_i$.

For a state $(S \subseteq F_j, b \in [0, \infty])$, we define $$f(S, b) := \alpha\sum_{i \in S} y_i d_i^p + (1 - y(S))b^p.$$
We assume $0 \cdot \infty^p = 0$.

Since Algorithm~\ref{alg:LMP} as described does not terminate in finite number of iterations with probability 1, we choose a parameter $T$ as the base case for our inductive proof; we shall let $T$ tend to $\infty$ later. We define an artificial cost $\Phi'_j$ of $j$ as follows. If $j$ is happily connected by the end of iteration $T$, then $\Phi'_j$ is the cost incurred when this happens. Otherwise, let $(S, b)$ be the state of $j$ at the end of iteration $T$, and we define $\Phi'_j := f(S, b)$. 

\begin{lemma}
\label{lemma:potential}
    Let $t \in [0, T]$. Suppose at the end of the $t$-th iteration, $j$ is not happily connected and its state is $(S, b)$. Conditioned on this event, we have $\E[\Phi'_j] \leq f(S, b)$.
\end{lemma}

Before we prove the lemma, we give some intuition behind the potential function $f(S, b)$. In the ideal case, we open at most one facility in $S$, and the probability we open $i \in S$ is $y'_i = y_i$. With the remaining probability of $1 - y(S)$, no facilities in $S$ is open and we use the backup connection distance $b$. In this case, the expected connection cost for $j$ would be $\sum_{i \in S} y_i d_i^p + (1 - y(S))b^p$. In the definition of $f(S, b)$, we lose a factor of $\alpha$ on $y_i d_i^p$ terms, but not on the $(1 - y(S))b^p$ term. This is reasonable as we always have the backup distance $b$.

\begin{proof}[Proof of Lemma~\ref{lemma:potential}]
    For convenience, we shall define $y := y(S)$, and $V: =\sum_{i \in S} y_i d_i^p$. So $f(S, b)$ is simply $\alpha V + (1 - y)b^p$. 

    We prove the lemma using induction over $t$ from $T$ down to $0$. The lemma clearly holds when $t = T$ by our definition of artificial cost. %
    Now we fix $t \in [1, T]$. We assume the lemma holds for the iteration $t$ and we prove that it holds for iteration $t - 1$. So, at the end of iteration $t-1$, which is the beginning of iteration $t$, the state of $j$ is $(S, b)$. In particular, $j$ is not happily connected yet.  The neighborhood graph $G = (F, E, w)$ and $y'$ are as defined at the beginning of iteration $t$. For convenience, we omit the subscript $G$ in $\delta^+$ and $\delta^-$. For any $i'i \notin E$, we let $w_{i'i} = 0$. 

    When we choose some facility $i' \in S \subseteq F_j$ in iteration $t$, $j$ will be happily connected during the iteration. Otherwise, let $(S' \subseteq S, b' \leq b)$ be the state at the end of the iteration. Let
    \begin{align*}
        x_{S'} &:= \sum_{i' \in F \setminus S}y'_{i'}\cdot \Pr[\text{the state becomes $(S', b')$ for some $b'$}|i'\text{ is chosen}],\\
        X &:= \sum_{S' \subseteq S} x_{S'} = \sum_{i' \in F \setminus S} y'_{i'},\\
        b_i&:= d_i + \min_{i' \in S: w_{i'i} < y_{i'}}d(i, i'), \forall i \in S, \quad \text{and}\quad  b_P:=\min_{i \in P}b_i,  \forall P \subseteq S.  
    \end{align*}
    \label{sec:connection_cost}

    Notice that when $i \notin S'$ (which implies $S'$ is well-defined and thus $j$ is not happily connected), then we have $b' \leq b_i$.
    This holds as for every $i' \in S$ with $w_{i'i} < y_{i'}$, and every $i'' \in F$ with $w_{i''i} > 0$, we have $d(i, i'') \leq d(i, i')$, and thus $d(j, i'') \leq d_i + d(i, i'') \leq d_i + d(i, i')$. So, we have $b' \leq \min\{b, b_{S \setminus S'}\}$.
    
    By the induction hypothesis, under the event of the lemma, we have 
    \begin{align}
        \E[\Phi'_j] \leq \frac{V + \sum_{S' \subseteq S} x_{S'} \cdot f\big(S', \min\{b, b_{S \setminus S'}\}\big)}{y + X}. \label{equ:cost-upper-bound}
    \end{align}
    The remaining goal of the proof is to upper bound \eqref{equ:cost-upper-bound} by $f(S, b)$.

We use $w(S, i)$ to denote the total $w$ value of the edges from $S$ to $i$, for every $i \in S$.  The numerator of \eqref{equ:cost-upper-bound} is
\begin{align}
&\quad V + \sum_{S'} x_{S'} \left(\alpha \sum_{i \in S'} y_i d_i^p + (1 - y(S'))\min\{b^p, b^p_{S \setminus S'}\}\right) \nonumber\\
&\leq V + X(1 - y)b^p + \sum_{S'}x_{S'}\left(\alpha \sum_{i \in S'} y_i d_i^p + \sum_{i \in S\setminus S'} y_i b_i^p\right) \label{inequ:1}\\
&= V + X(1 - y) b^p + \sum_{i \in S}y_i \left(\alpha d_i^p\sum_{S' \ni i}x_{S'} + b_i^p \sum_{S' \not\ni i} x_{S'}\right) \nonumber\\
&=V + X(1 - y) b^p + \sum_{i \in S}y_i \big(\alpha (X - 1 + w(S, i))d_i^p  + (1 - w(S, i))b_i^p\big) \label{inequ:2}\\
&\le V +X(1 - y) b^p  + \alpha(X-1)V  + y(1 - y)b^p+ \sum_{i \in S}y_i \big(\alpha w(S, i)d_i^p  + (y - w(S, i))b_i^p \big) \nonumber\\
&\leq V + (X+y)(1 - y) b^p + \alpha(X-1)V +\sum_{i\in S}y_i\Bigg(\sum_{i' \in S} w_{i'i} \alpha d_i^p + \sum_{i' \in S} (y_{i'} - w_{i'i})(d_i + d(i,i'))^p \Bigg)\label{inequ:3}\\
&\le V + (X+y)(1 - y) b^p +  \alpha(X-1)V+ \sum_{i\in S}\sum_{i'\in S}y_iy_{i'}\max\{\alpha d_i^p, (d_i + d(i,i'))^p\}\nonumber\\
&\le V + (X+y)(1 - y) b^p + \alpha(X-1)V + (2\alpha-1)yV \label{inequ:apply-alpha}\\
&= (X+y)(1 - y) b^p + \left(1+\alpha(X-1) + (2\alpha-1)y'\right) V \nonumber \\
&\leq (X+y)(1 - y)b^p + (\alpha X + \alpha y)V \quad =\quad (X+y)f(S, b). \label{inequ:4}
\end{align}

To see \eqref{inequ:1}, notice that for every $S' \subseteq S$ we have $(1 - y(S'))\min\{b^p, b^p_{S \setminus S'}\} \leq (1 - y)b^p + (y - y(S')) b^p_{S \setminus S'} \leq (1 - y)b^p + \sum_{i \in S \setminus S'} y_i b_i^p$. For \eqref{inequ:2}, notice that $\sum_{S' \not\ni i} x_{S'} = 1 - w(S, i)$, and $\sum_{S' \ni i} x_{S'} = X - 1 + w(S, i)$. \eqref{inequ:3} used that $b_i \leq d_i + d(i, i')$ for every $i' \in S$ with $w_{i'i} < y_{i'}$. \eqref{inequ:apply-alpha} used \eqref{equ:distance_sum_goal}. Finally, \eqref{inequ:4} used that $1 - \alpha + (\alpha-1)y = (1-\alpha)(1 - y) \leq 0$.
\end{proof}

Applying Lemma~\ref{lemma:potential} to the end of iteration $0$, when the state of $j$ is $(F_j, \infty)$,  we obtain that $\E[\Phi'_j] \leq f(F_j, \infty) = \alpha \sum_{i \in F_j} y_i d_i^p$.  Notice that when $T$ tends to $\infty$, we have $\E[\Phi'_j]$ tends to $\E[\Phi_j]$, as the probability that $j$ is not happily connected with $T$ iterations tends to $0$. This proves Theorem~\ref{thm:cost-analysis}.

\subsection{$L_p^p$ objective in general metrics}
\label{subsec:LMP-general-metric}

\begin{lemma}
    \label{lemma:distance_sum_Lp}
    In any metric space and for any $p\ge 1$ and $\alpha=\frac{3^p+1}{2}$, we have that \eqref{equ:distance_sum_goal} holds.
\end{lemma}

\begin{proof}
    We have that
    \begin{align*}
        &\sum_{i\in T,i'\in T}y_iy_{i'}\cdot \max\{\alpha d^p(i,j),(d(i,j)+d(i,i'))^p\} \\
        \le{}&\sum_{i\in T,i'\in T}y_iy_{i'}\cdot \max\{\alpha d^p(i,j),(d(i,j)+d(i,j)+d(i',j))^p\} \\
        \le{}&\sum_{i\in T,i'\in T}y_iy_{i'}\cdot \max\{\alpha d^p(i,j),3^{p-1}(d^p(i,j)+d^p(i,j)+d^p(i',j))\} \tag{convexity of $x^p$}\\
        \le{}&\sum_{i\in T,i'\in T}y_iy_{i'}\cdot (2\cdot 3^{p-1}\cdot d^p(i,j)+3^{p-1}\cdot d^p(i',j)) \\
        ={}&3^py(T)\cdot\sum_{i\in T}y_id^p(i,j). \qedhere
    \end{align*}
\end{proof}

\subsection{Euclidean $k$-means}

\begin{theorem}
    \label{thm:distance_sum_euclid_kmeans}
    Suppose that the metric space is  Euclidean, $p = 2$ and $\alpha=4$, then \eqref{equ:distance_sum_goal} holds.
\end{theorem}

For every facility $i\in T$, let $\vec{v}_i$ denote the vector corresponding to $i$ in the Euclidean space, and assume that the client $j$ is at $\vec{0}$. For notational simplicity, let $d_i=|\vec{v}_i|$ and let $I_{i,i'}=\langle \vec{v}_i,\vec{v}_{i'}\rangle$.  Then, \eqref{equ:distance_sum_goal} can be written as
\begin{align*}
    \sum_{i\in T,i'\in T}y_iy_{i'}\cdot \max\left\{\alpha d_i^2,\left(d_i+\sqrt{d_i^2+d_{i'}^2-2I_{i,i'}}\right)^2\right\}\le (2\alpha-1)y(T)\sum_{i\in T}y_id_i^2.
\end{align*}
We move the terms in the RHS into the summation of the LHS:
\begin{align}
    \label{equ:distance_sum_goal_euclid_kmeans}
    \sum_{i\in T,i'\in T}y_iy_{i'}\cdot \bk*{\max\left\{\alpha d_i^2,\left(d_i+\sqrt{d_i^2+d_{i'}^2-2I_{i,i'}}\right)^2\right\}-(\alpha-1/2)(d_i^2+d_{i'}^2)}\le 0.
\end{align}
Next, we study the difference
\begin{align*}
    \text{diff}(i,i')=\max\left\{\alpha d_i^2,\left(d_i+\sqrt{d_i^2+d_{i'}^2-2I_{i,i'}}\right)^2\right\}-(\alpha-1/2)(d_i^2+d_{i'}^2).
\end{align*}
More precisely, we study a pair of differences $\text{diff}(i,i')+\text{diff}(i',i)$. We will show that
\begin{claim}
    \label{clm:linear_bound_euclid}
    For any $i,i'\in T$, we have that $\text{diff}(i,i')+\text{diff}(i',i)\le -6\cdot I_{i,i'}=-3\cdot I_{i,i'}-3\cdot I_{i',i}$.
\end{claim}
Note that this holds even when $i=i'$. This implies \cref{thm:distance_sum_euclid_kmeans}.

\begin{proof}[Proof of \cref{thm:distance_sum_euclid_kmeans} using \cref{clm:linear_bound_euclid}]
    
    Combining \cref{clm:linear_bound_euclid} with the fact that
    \begin{align*}
        \sum_{i\in T,i'\in T}y_iy_{i'}I_{i,i'}=|\sum_{i\in T}y_i\vec{v}_i|^2\ge 0,
    \end{align*}
    we can bound the LHS of \eqref{equ:distance_sum_goal_euclid_kmeans} as
    \begin{align*}
        \text{LHS of }\eqref{equ:distance_sum_goal_euclid_kmeans}&\le \sum_{i\in T,i'\in T}y_iy_{i'}\cdot(-3I_{i,i'})\le 0.\qedhere
    \end{align*}
\end{proof}

It remains to prove \cref{clm:linear_bound_euclid}.
\begin{proof}[Proof of \cref{clm:linear_bound_euclid}]
    We first relax $\text{diff}(i,i')$ as %
    \begin{align*}
        \text{diff}(i,i')\le \max\{\alpha d_i^2,2d^2_i+2\bk*{d_i^2+d_{i'}^2-2I_{i,i'}}\}-(\alpha-1/2)(d_i^2+d_{i'}^2).\tag{$(a+b)^2\le 2a^2+2b^2$}
    \end{align*}
    To upper bound $\text{diff}(i,i')+\text{diff}(i',i)$, we consider the following cases, based on which terms in the maximums are chosen:
    \begin{enumerate}
        \item Both maximums choose the first term: In this case, we need to show that
        \begin{align*}
            \alpha d_i^2+\alpha d_{i'}^2-(2\alpha-1)(d_i^2+d_{i'}^2)\le -6I_{i,i'}.
        \end{align*}
        When $\alpha=4$, the LHS is
        \begin{align*}
            -3(d_i^2+d_{i'}^2)\le -3\cdot (2d_id_{i'})\le -6I_{i,i'}.
        \end{align*}
        \item Exactly one of maximums choose the first term (w.l.o.g. assume that the maximum in $\text{diff}(i,i')$ chooses the first term): In this case, we need to show that
        \begin{align*}
            2d_i^2+2(d_i^2+d_{i'}^2-2I_{i,i'})+\alpha d_{i'}^2-(2\alpha-1)(d_i^2+d_{i'}^2)\le -6I_{i,i'}.
        \end{align*}
        When $\alpha=4$, the LHS is
        \begin{align*}
            &-3d_i^2-d_{i'}^2-4I_{i,i'} \\
            \le{}&-d_i^2-d_{i'}^2-4I_{i,i'} \\
            \le{}&-2d_id_{i'}-4I_{i,i'}\le -6I_{i,i'}.\tag{$I_{i,i'}\le d_id_{i'}$}
        \end{align*}
        \item Both maximums choose the second term: In this case, we need to show that
        \begin{align*}
            2d_i^2+2(d_i^2+d_{i'}^2-2I_{i,i'})+2d_{i'}^2+2(d_i^2+d_{i'}^2-2I_{i,i'})-(2\alpha-1)(d_i^2+d_{i'}^2)\le -6I_{i,i'}.
        \end{align*}
        When $\alpha=4$, the LHS is
        \begin{align*}
            &-8I_{i,i'}-d_i^2-d_{i'}^2 \\
            \le{}& -8I_{i,i'}-2d_id_{i'}\le -6I_{i,i'}.\tag{$I_{i,i'}\ge -d_id_{i'}$}
        \end{align*}
    \end{enumerate}
    This concludes the proof of the claim.
\end{proof}

\subsection{Better analysis for Euclidean $k$-means}

By removing the relaxation in the proof of \cref{clm:linear_bound_euclid}, we can obtain a better bound (with a more complicated proof):

\begin{theorem}
    \label{thm:distance_sum_euclid_kmeans_better}
 Suppose that the metric space is Euclidean, $p = 2$ and $\alpha=11/3$, then \eqref{equ:distance_sum_goal} holds.
\end{theorem}

We only need to prove a better bound for $\text{diff}(i,i')+\text{diff}(i',i)$:

\begin{claim}
    \label{clm:linear_bound_euclid_better}
    When $\alpha=11/3$, for any $i,i'\in T$, we have that $\text{diff}(i,i')+\text{diff}(i',i)\le -2(\alpha-1)I_{i,i'}$.
\end{claim}

\begin{proof}
    Without loss of generality, we only prove the claim in the case where $d_id_{i'}=1$, in which case $I_{i,i'}\in [-1,1]$. For notational simplicity, we use $x$ to denote $d_i$ (in which case $d_{i'}=1/x$), and use $I$ to denote $I_{i,i'}$.
    
    Similar to \cref{clm:linear_bound_euclid}, we prove the claim by considering three cases, based on which term in the maximum is chosen.

    \begin{enumerate}
        \item Both maximums choose the first term: In this case, we need to show that
        \begin{align*}
            \alpha x^2+\alpha (1/x)^2-(2\alpha-1)(x^2+(1/x)^2)\le -2(\alpha-1)I,
        \end{align*}
        which holds for any $\alpha\ge 1$ since $2I\le (x^2+(1/x)^2)$.
        \item One of the maximums choose the first term: In this case, we need to prove the following claim.

        \begin{claim}
            For any $x>0$ and any $I\in [-1,1]$, we have that
            \begin{align}
                \label{equ:diff_goal_1}
                \alpha x^2+\bk*{(1/x)+\sqrt{x^2+(1/x)^2-2I}}^2-(2\alpha-1)(x^2+(1/x)^2)\le -2(\alpha-1)I.
            \end{align}
        \end{claim}
        \begin{proof}
            Let $w=\sqrt{x^2+(1/x)^2-2I}$. We can rewrite \eqref{equ:diff_goal_1} as
            \begin{align*}
                \alpha x^2+((1/x)+w)^2-(2\alpha-1)(x^2+(1/x)^2)\le (\alpha-1)w^2-(\alpha-1)\cdot (x^2+(1/x)^2).
            \end{align*}
            After rearranging the terms, this is equivalent to showing that
            \begin{align*}
                (\alpha-2)w^2-2(1/x)\cdot w+(\alpha-1)(1/x)^2\ge 0.
            \end{align*}
            Plugging in $\alpha=11/3$, we need to show that
            \begin{align*}
                5w^2-6(1/x)\cdot w+8(1/x)^2\ge 0.
            \end{align*}
            This holds because
            \begin{align*}
                &5w^2-6(1/x)\cdot w+8(1/x)^2 \\
                ={}&5(w-(3/5)(1/x))^2+(31/5)(1/x)^2\ge 0.\qedhere
            \end{align*}
        \end{proof}
        \item Both maximums choose the second term: In this case, we need to prove the following claim.
        \begin{claim}
            For any $x>0$ and any $I\in [-1,1]$, we have that
            \begin{align}
                \label{equ:diff_goal_2}
                &\bk*{x+\sqrt{x^2+(1/x)^2-2I}}^2& \nonumber\\
                +&\bk*{(1/x)+\sqrt{x^2+(1/x)^2-2I}}^2& \nonumber\\
                -&(2\alpha-1)(x^2+(1/x)^2)&\le -2(\alpha-1)I.
            \end{align}
        \end{claim}
        \begin{proof}
            After expanding the terms, the LHS is equal to
            \begin{align*}
                x^2+(1/x)^2+2(x+(1/x))\sqrt{x^2+(1/x)^2-2I}+2(x^2+(1/x)^2-2I)-(2\alpha-1)(x^2+(1/x)^2).
            \end{align*}
            Letting $v=x^2+1/x^2$ (note that when $x>0$, we have $v\ge 2$), this can be simplified to
            \begin{align*}
                2\sqrt{v+2}\cdot\sqrt{v-2I}+(4-2\alpha)v-4I.
            \end{align*}
            It remains to show that
            \begin{align*}
                2\sqrt{v+2}\cdot\sqrt{v-2I}+(4-2\alpha)v-4I\le -2(\alpha-1)I.
            \end{align*}
            Plugging in $\alpha=11/3$, we need to show that
            \begin{align*}
                \sqrt{v+2}\cdot \sqrt{v-2I}\le -(2/3)I+(5/3)v.
            \end{align*}
            Note that since $v\ge 2$ and $I\in [-1,1]$, both sides are positive. Therefore, we can take squares of both sides:
            \begin{align*}
                v^2+(2-2I)v-4I\le (25/9)v^2-(20/9)vI+(4/9)I^2.
            \end{align*}
            Rearranging the terms gives
            \begin{align*}
                (16/9)v^2-(2/9)vI-2v+(4/9)I^2+4I\ge 0.
            \end{align*}
            When $I$ is fixed, the derivative of the LHS w.r.t. $v$ is
            \begin{align*}
                (32/9)v-(2/9)I-2,
            \end{align*}
            which is always positive when $v\ge 2$ and $I\in [-1,1]$. Thus, we only have to consider the case where $v=2$. In this case, we need to show that
            \begin{align*}
                &(64/9)-(4/9)I-4+(4/9)I^2+4I \\
                ={}&(4/9)\cdot (I^2+8I+7)\ge 0,
            \end{align*}
            which holds when $I\in [-1,1]$.
        \end{proof}
    \end{enumerate}
    This concludes the proof of \cref{clm:linear_bound_euclid_better}.
\end{proof}

%% file: ProcessingLP.tex
\section{Making the $y$ values in the LP solution integer multiples of $\varepsilon^{12p^2}$}
\label{sec:preprocessing}

From this section to Section~\ref{sec:cost-new}, we shall describe and analyze our iterative rounding algorithm for the $k$-clustering problem, that opens $k + O_{\varepsilon, p}(1)$ facilities with large probability.  We assume $\varepsilon<\frac{1}{3p^4}$ and $\frac1\varepsilon$ is an integer.  For the sake of notational convenience, we allow us to lose an additive factor of $2^{O(p)} \varepsilon$ in the approximation ratio. To convert this loss to $\varepsilon$, we can scale $\varepsilon$ down by $2^{O(p)}$ at the beginning.

In this section, we show how to preprocess the solution $(x, y)$ obtained from solving the LP to another LP solution $(x'', y'')$ whose coordinates are integer multiples of $\varepsilon^{12p^2}$, with a small loss in the cost of the solution.\footnote{$\varepsilon^{12p^2}$ should be treated as $\varepsilon^{\lceil12p^2\rceil}$ in case $p$ is not an integer.}

For every point $j \in F\cup C$, we define $F_j$ to be the set of facilities closest to $j$ with total fractional value equal 1.
By splitting facilities, we assume for every $i \in F$, the whole $y_i$ fraction of $i$ is either completely inside $F_j$, or completely outside. That is, we have $y(F_j) = 1$ and $\max_{i \in F_j}d(j, i) \leq \min_{i \in F \setminus F_j} d(j, i)$ for every $j \in C$. %

Overall, our algorithm contains the following steps:
\begin{enumerate}[topsep=3pt, itemsep=0pt]
    \item We first use a filtering step to create a set $C^*$ of representatives, each $j \in C^*$ with a ball $B_j \subseteq F_j$ of facilities centered at $j$, so that $B_j$'s for $j \in C^*$ are disjoint. 
    \item For every $j \in C^*$, we define a small ``core'' $B'_j \subseteq B_j$ around the center $j$. We only keep one fractional facility in $B'_j$, chosen with probabilities proportional to $y_i$ values. Let $y'$ be the new fractional opening vector. 
    \item Finally, we randomly round each $y'_i$ value up or down to the nearest integer multiple of $\varepsilon^{12}$, while maintaining the sum of $y'_i$ values for each ball $B_j$, and the whole set $F$.  This gives our final fractional opening vector $y''$.
\end{enumerate}

Till the end of the paper, for a given fractional opening vector $\bar y \in [0, 1]^F$ and a client $j \in C$, we shall use $\cost_{\bar y}(j)$ to denote the cost of $j$ in the solution $\bar y$, obtained by connecting $j$ to the nearest 1 fractional facility. Formally, it is the minimum of $\sum_{i \in F} \bar x_{ij} d^p(i, j)$ subject to $\bar x_{ij} \in [0, \bar y_i]$ for every $i \in F$, and $\sum_{i \in F} \bar x_{ij} = 1$. Accordingly, for $y'$ and $y''$, let $x'_{ij}$ and $x''_{ij}$ respectively denote the optimal fractional connection variables that achieve $\cost_{y'}(j)$ and $\cost_{y''}(j)$ for each client $j \in C$.

\subsection{Filtering to choose a set $C^*$ of representatives}

We define $d_\av(j)$ and $d_{\max}(j)$ for every point $j \in C$ as follows:
\begin{align*}
	d_\av(j) := \sum_{i \in F_j} y_i \cdot d(i, j) \qquad \text{and}\qquad d_{\max}(j):=\max_{i \in F_j} d(i, j).
\end{align*}

\begin{algorithm}[h]
    \caption{Filtering}
    \label{alg:filtering}
    $C' \gets C, C^* \gets \emptyset$\\
    \While{$C' \neq \emptyset$}{ \label{step:filtering-loop}
        $j^* \gets$ client in $C'$ with the smallest $d_\av(j^*)$\\
        $C^* \gets C^* \cup \{j^*\}$\\
        remove from $C'$ all clients $j$ with $d(j, j^*) \leq \left(\frac{1}{(\varepsilon/p)^{4p}}-2\right)d_\av(j)$ (including $j^*$ itself)
    }
    \Return $C^*$\;
\end{algorithm}

We use Algorithm~\ref{alg:filtering} to obtain a set $C^*$ of representatives. For every $j \in C^*$, if $d(j, C^* \setminus \{j\}) / 2 < d_{\max}(j)$, then we define $B_j = \ball^\circ_F(j, d(j, C^* \setminus \{j\}) / 2 ) \subsetneq F_j$; otherwise we define $B_j = F_j$.  Clearly, the balls $B_j$ for $j \in C^*$ are disjoint.

\begin{claim}
    \label{claim:Bj-volume}
	For $\varepsilon < \frac{1}{2}$ we have 
    $ \frac{1}{2-\varepsilon}  < y(B_j) \leq 1$ for every $j \in C^*$. 
\end{claim}

To see the above claim, recall that $d(j, C^* \setminus \{j\}) > \left(\frac{1}{(\varepsilon/p)^{4p}}-2\right)d_\av(j)$. Then, by Markov inequality, there is at least $\frac{1}{2- \varepsilon}$ of fractional facility opening within distance $\frac{1}{1-\frac{1}{2-\varepsilon}}\cdot d_\av(j) < \frac{\frac{2-\varepsilon}{1-\varepsilon}\cdot d(j, C^* \setminus \{j\})}{\frac{1}{(\varepsilon/p)^{4p}}-2} < \frac{ d(j, C^* \setminus \{j\})}{2}$.

\begin{definition}
	For $j \in C$, we define its representative as the $j^*$ we chose in the iteration of the loop \ref{step:filtering-loop} of \cref{alg:filtering} where $j$ is removed from $C'$. 
\end{definition}
Clearly, for the representative $j'$ of $j$, we have $d_\av(j') \leq d_\av(j)$ and $d(j, j') \leq \left(\frac{1}{(\varepsilon/p)^{4p}}-2\right)d_\av(j)$. 
Notice that it is possible the representative of $j$ is itself, in case $j \in C^*$.  We say a client $j$ is of
\begin{itemize}
	\item type-1  if $d_{\max}(j) \leq \frac{1}{(\varepsilon/p)^{4p}} \cdot d_\av(j)$,
	\item type-2 if it is not of type-1, and its representative $j' \in C^*$ has $B_{j'} = F_{j'}$, and
	\item type-3 otherwise.
\end{itemize}

\subsection{Rounding $y$ to $y'$}

For every $j \in C^*$, we define $B'_{j}:=\ball_{F}(j, \varepsilon d_{\max}(j))$. We prove in Lemma~\ref{lem:disjoint_balls} that $B'_{j} \subseteq B_{j}$; so the balls $B'_{j}$ for $j \in C^*$ are also disjoint. 

We round $y$ to $y'$ as follows. For every $j' \in C^*$, we randomly choose a facility $i^* \in B'_{j'}$ with probabilities proportional to $y_{i^*}$'s. We set $y'_{i^*} = y(B'_{j'})$; for all facilities $i \in B'_{j'} \setminus \{i^*\}$, we set $y'_{i} = 0$. For all the facilities $i$ not in any ball $B'_{j'}$ for any $j'\in C^*$, we set $y'_i = y_i$. We remark that this step is needed when we analyze the cost of type-3 clients later.

\begin{claim}
    \label{claim:y-to-y-expectation}
    For every facility $i \in F$, we have $\E[y'_i] = y_i$.
\end{claim}

\begin{lemma}
\label{lem:disjoint_balls}
 For every $j' \in C^*$, we have $B'_{j'} \subseteq B_{j'}$.
\end{lemma}

\begin{proof}
    The lemma holds trivially if $B_{j'} = F_{j'} \supseteq \ball^\circ_F(j', d_{\max}(j'))$. So we assume $B_{j'} \neq F_{j'}$. Let $j''$ be the nearest neighbor of $j'$ in $C^*$. So we have $B_{j'} = \ball^\circ_F(j, d(j', j'')/2)$, and it remains to prove $\varepsilon d_{\max}(j') \leq d(j', j'')/2$. 

Define $S_{j'} := \ball_F(j', 2d_{\av}(j'))$ and $S_{j''} = \ball_F(j'', 2d_{\av}(j''))$. By Markov's inequality, balls $S_{j'}$ and $S_{j''}$ each contain at least $1/2$ fractional value. By the filtering condition, we have $d(j', j'') > (\frac{1}{(\varepsilon/p)^{4p}} - 2) \max(d_{\av}(j'), d_{\av}(j'')) > 2(d_{\av}(j') + d_{\av}(j''))$, which ensures $S_{j'}$ and $S_{j''}$ are disjoint and implies $y(S_{j'}\cup S_{j''}) \ge 1$. Therefore, we have $d_{\max}(j') \le \max\{2d_\av(j'), d(j', j'') + 2d_{\av}(j'')\}$. In case $d_{\max}(j') \leq 2d_\av(j')$, we would have $B_{j'} = F_{j'}$. So, $d_{\max}(j') \leq d(j', j'') + 2d_{\av}(j'')$. Using the filtering condition, we have
\begin{align*}
d_{\max}(j') \le d(j', j'') + 2d_{\av}(j'') < d(j', j'') + \frac{2(\varepsilon/p)^{4p}}{1-2(\varepsilon/p)^{4p}} d(j', j'') = \left(1 + \frac{2(\varepsilon/p)^{4p}}{1-2(\varepsilon/p)^{4p}}\right) d(j', j''),
\end{align*}
which implies $\varepsilon d_{\max}(j') < \varepsilon\left(1 + \frac{2(\varepsilon/p)^{4p}}{1-2(\varepsilon/p)^{4p}}\right) d(j',j'') < d(j',j'')/2$.
\end{proof}

\begin{lemma}
    \label{lemma:y-to-y'}
    For every $j \in C$, we have 
    $\E[\cost_{y'}(j)] \leq (1 + O(p\varepsilon)) \cost_{y}(j)$.
\end{lemma}
\begin{proof}
    Consider a representative $j'$. We consider how the rounding for $B'_{j'}$ affects the cost of $j$. If $d(j, j') \leq d_{\max}(j')/3$, then $d_{\max}(j) \geq d_{\max}(j') - d(j, j') \geq 2d_{\max}(j')/3$. The distance between $j$ and any facility in $B'_{j'}$ is at most $d(j, j') + \varepsilon d_{\max}(j') \leq (\frac{1}{3} + \varepsilon) d_{\max}(j')$. Therefore, the ball $B'_{j'}$ is completely inside $F_j$. Since $i^*$ is sampled exactly with probability $\frac{y_{i^*}}{y(B'_{j'})}$, its expected distance is:
    \begin{align*}
        \mathbb{E}[y'(B'_{j'})d^p(j,i^*)] = \sum_{i\in B'_{j'}} y(B'_{j'}) \cdot \frac{y_i}{y(B'_{j'})} d^p(j, i) = \sum_{i \in B'_{j'}} y_i d^p(j, i).
    \end{align*}
    The expected connection cost is exactly equal to the original cost of $B'_{j'}$, without any loss.

    Consider the other case $d(j, j') > d_{\max}(j')/3$. The ratio between the distances of any two facilities in $B'_{j'}$ to $j$ is at most $\frac{1/3 + \varepsilon}{1/3 - \varepsilon} = 1 + O(\varepsilon)$. 
    Therefore, the rounding for $B'_{j'}$ only incur an $(1 + O(p\varepsilon))$ factor in the cost associated with $B_{j'}$.
    
    Combining the two cases proves the lemma.
\end{proof}

\subsection{Rounding $y'$ to $y''$}

We define a laminar family $\mathcal{S} := \{\{i\}: i \in F\} \cup  \{B_j: j \in C^*\} \cup \{\{F\}\}$.  Then, we perform randomized pipage rounding (also known as dependent rounding) from~ \cite{srinivasan2001distributions} guided by family $\mathcal{S}$, on $y'/\varepsilon^{12p^2}$ and return the scaled back vector $y''$. The algorithm can be implemented to approximately preserve sums of entries within subsets from $\mathcal{S}$, see e.g.~\cite{byrka2010fault}. 

\begin{lemma} \label{lem:dep_rounding}
The randomized fractional solution $y'' \in [0, 1]^F$ obtained from $y' \in [0, 1]^F$ satisfies:
\begin{enumerate}
\item
\label{y-double-prime-expectation}
$\E[y''_i] = y'_i,\quad \forall i \in F$;
\item \label{approx_sum_preservation}	$
    y''(S)/\varepsilon^{12p^2} \in \{ \lfloor y'(S)/\varepsilon^{12p^2} \rfloor, \lceil y'(S)/\varepsilon^{12p^2} \rceil \}, \quad \forall S \in \mathcal{S};
    $
\item \label{total_variance} $\Var[y''(F_j)]\leq \sum\limits_{i\in F_j}\Var[y''_i], \quad \forall j \in C$.
\end{enumerate}    
\end{lemma} 

\begin{proof}
    The properties of the rounding procedure described in~\cite{srinivasan2001distributions} hold for an arbitrary order in which fractional entries are paired for rounding. Given a laminar family of subsets we fix the order of rounding so that entries from within a subset are paired up together for rounding until there only exists one fractional element that can be rounded together with an element outside of the set. It remains to observe that family $\{\{i\}: i \in F\} \cup  \{B_j: j \in C^*\} \cup \{F\}$ is laminar to see that this way we obtain Property~\ref{approx_sum_preservation}.

    Obtaining tail bounds from negative correlation is common in the literature, see e.g.~\cite{panconesi1997randomized}. In this argument we choose to directly argue about the variance:
    \[
    \Var[y''(F_j)] = \sum\limits_{i,i'\in F_j}\Cov[y''_i,y''_{i'}] = \sum\limits_{i\in F_j}\Var[y''_i] + \sum\limits_{i,i'\in F_j, i\neq i'}\Cov[y''_i,y''_{i'}].
    \]
    To see that Property~\ref{total_variance} holds it suffices to see that Property A3 from~\cite{srinivasan2001distributions} implies $\Cov[y''_i,y''_{i'}] \leq 0$ for all $i,i'\in F_j, i\neq i'$.
\end{proof}

\begin{lemma}
    \label{lemma:type-1-cost}
	For a type-1 client $j \in C$, we have 
	\begin{align*}
		\E[\cost_{y''}(j)] \leq \left(1 + O\left(\varepsilon^{p^2}\right)\right) \cost_{y}(j).
	\end{align*}
\end{lemma}
\begin{proof} %
	We first show that $y''\left(\ball_F\left(j,  O\big(\frac1{(\varepsilon/p)^{4p}}\big) d_\av(j)\right)\right) \geq 1$ happens with probability 1, for some large enough hidden constant in $O\big(\frac1{(\varepsilon/p)^{4p}}\big)$.  Consider the representative $j'$ of $j$. Notice that $d_{\max}(j') \leq d(j, j') + d_{\max}(j) \leq O\big(\frac1{(\varepsilon/p)^{4p}}\big)\cdot d_\av(j)$. If $B_{j'} = F_{j'}$, then we have $y''(B_{j'}) = 1$ and so the statement holds. Otherwise, let $j''$ be the nearest neighbor of $j'$ in $C^* \setminus \{j'\}$. Then $d(j', j'') < 2 d_{\max}(j')$, $B_{j'} = \ball^\circ_F(j', d(j', j'')/2)$, and $B_{j''} \subseteq \ball^\circ_F(j'', d(j', j'')/2)$.  Therefore, we have $y''(B_{j'} \cup B_{j''}) > 1$ by \cref{claim:Bj-volume}.  Every facility in $B_{j'} \cup B_{j''}$ has distance at most $d(j, j') + d(j', j'') + d(j', j'')/2 \leq d(j, j') + 3d_{\max}(j') \leq O\big(\frac1{(\varepsilon/p)^{4p}}\big) d_\av(j)$ from $j$.  \medskip
	
	Then, 
	\begin{align}
		\E[\cost_{y''}(j)] &\leq \E\left[\sum_{i \in F_j} y''_i d^p(j, i) + (1 - y''(F_j))_+ \cdot \left(O\left(\frac1{(\varepsilon/p)^{4p}}\right) \cdot d_\av(j)\right)^p\right] && \nonumber\\
        &= \cost_{y}(j) + \E[(1 - y''(F_j))_+] \cdot O\left(\frac1{(\varepsilon/p)^{4p^2}}\right) \cdot d_\av^p(j) && \label{ineq:type1-1}\\
        &\le \cost_{y}(j) + \E[(1 - y''(F_j))_+] \cdot O\left(\frac1{(\varepsilon/p)^{4p^2}}\right) \cdot \cost_{y}(j).\label{ineq:type1-2}
	\end{align}

    \eqref{ineq:type1-1} follows from the Property~\ref{y-double-prime-expectation} of Lemma~\ref{lem:dep_rounding} and Claim~\ref{claim:y-to-y-expectation}, which ensure $\E[y''_i] = y'_i$ and $\E[y'_i] = y_i$. For \eqref{ineq:type1-2}, notice that $d_\av^p(j)=\left(\sum_{i\in F_j}y_i d(i,j)\right)^p\le \sum_{i\in F_j}y_id^p(i,j)= \cost_{y}(j)$.
    
	It remains to bound $\E[(1 - y''(F_j))_+]$. Applying Cauchy-Schwarz inequality, we obtain
    
	$$
    \E[(1 - y''(F_j))_+]\le \E[|1-y''(F_j)|]\le\sqrt{\E[(1-y''(F_j))^2]}.
    $$

    We can show that $\E[y''(F_j)]=\sum_{i\in F_j}\E[y''_i]=\sum_{i\in F_j}\E[y'_i]=1$, so $\E[1-y''(F_j)]=0$, and we have

    \begin{align*}
    \E[(1-y''(F_j))^2]= \Var[1-y''(F_j)]
    = \Var[y''(F_j)]
    \leq \sum_{i\in F_j}\Var[y''_i].
    \end{align*}

    For every $i\in F_j$, let $z'_i=y'_i-\lfloor \frac{y'_i}{\varepsilon^{12p^2}}\rfloor \varepsilon^{12p^2}$, and $z''_i=\varepsilon^{12p^2}$ if $y''_i/\varepsilon^{12p^2}=\lfloor y'_i/\varepsilon^{12p^2}\rfloor+1$, $z''_i=0$ if $y''_i/\varepsilon^{12p^2}=\lfloor y'_i/\varepsilon^{12p^2}\rfloor$, then we have
    \begin{align*}
    \Var[y''_i]&=\Var[z''_i]
    =\E[{z''_i}^2]-\E[z''_i]^2
    =\left(\frac{z_i'}{\varepsilon^{12p^2}}\right)(\varepsilon^{12p^2})^2-z_i'^2
    \leq z'_i\varepsilon^{12p^2}.
    \end{align*}

    Summing this over all $i \in F_j$ yields
    
    \begin{align*}
    \E[(1-y''(F_j))^2]\le \sum_{i\in F_j}\Var[y''_i]
    \le \sum_{i\in F_j}z'_i\varepsilon^{12p^2}
    \le \sum_{i\in F_j}y'_i\varepsilon^{12p^2}
    = \varepsilon^{12p^2}.
    \end{align*}

    So we can obtain that $\E[(1 - y''(F_j))_+]\le \sqrt{\E[(1-y''(F_j))^2]}\le \sqrt{\varepsilon^{12p^2}}=\varepsilon^{6p^2}$.

    Therefore,
	we have 
	\begin{flalign*}
		\E[\cost_{y''}(j)] \leq \cost_{y}(j) + \E[(1 - y''(F_j))_+] \cdot O\left(\frac1{(\varepsilon/p)^{4p^2}}\right) \cdot \cost_{y}(j) = \left(1 + O\left(\varepsilon^{p^2}\right)\right) \cdot \cost_{y}(j).
	\end{flalign*}

    We use $\varepsilon<\frac{1}{3p^4}$ in the last step.
\end{proof}

\begin{lemma} \label{lem:not-type1-distance}
	If $j \in C$ is not of type-1, and $j'$ is the representative of $j$, then $d(j, j') \leq 4d_\av(j)$. 
\end{lemma}
\begin{proof}
	As $j$ is not of type-1, we have  $d_{\max}(j) > \frac{1}{(\varepsilon/p)^{4p}} \cdot d_\av(j)$. Assume towards the contradiction that $d(j, j') > 4 d_\av(j)$.  By the properties of representative, we have $d(j, j') \leq \left(\frac{1}{(\varepsilon/p)^{4p}} - 2\right)\cdot d_\av(j)$. 
	
	Notice that $y(\ball_F(j, 2d_\av(j))) > 1/2$ and  $y(\ball_F(j', 2d_\av(j'))) > 1/2$.  The two balls are disjoint, and are both inside $\ball_F(j, d_{\max}(j))$. This contradicts with the definition of $F_j$ and $d_{\max}(j)$. 
\end{proof}

\begin{lemma} \label{lem:type-2-cost}
For a type-2 client $j \in C$, we have $\mathbb{E}[\cost_{y''}(j)] \le (1 + O(\varepsilon))\E[\cost_{y'}(j)]$.
\end{lemma}
\begin{proof}
We have $d_{\max}(j) > \frac{1}{(\varepsilon/p)^{4p}} \cdot d_\av(j)$ and the representative $j'$ of $j$ has $B_{j'} = F_{j'}$. Notice that we always have $y''(B_{j'}) = 1$. So,

\begin{align*}
\E\left[\cost_{y''}(j)\right]&\le \mathbb{E}\left[\sum_{i \in F_{j'}} y''_i d^p(j, i)\right] = \E\left[\sum_{i \in F_{j'}} y'_i d^p(j, i)\right] \\
&= \E[\cost_{y'}(j)] + \E\left[\sum_{i \in F_{j'}} (y'_i - x'_{ij}) d^p(j, i) - \sum_{i \notin F_{j'}} x'_{ij} d^p(j, i)\right].
\end{align*}
For any $i \in F_{j'}$, we have $d(j, i) \le d(j, j') + d_{\max}(j')$. For any $i \notin F_{j'}$, we have $d(j', i) > d_{\max}(j')$, which implies $d(j, i) > d_{\max}(j') - d(j, j')$. Since $\sum_{i \in F_{j'}} (y'_i - x'_{ij}) = \sum_{i \notin F_{j'}} x'_{ij}$, the expected cost is bounded by
\begin{align*}
& \E[\cost_{y'}(j)] + \left(\left(\frac{d_{\max}(j') + d(j, j')}{d_{\max}(j') - d(j, j')}\right)^p - 1\right) \cdot \E\left[\sum_{i \notin F_{j'}} x'_{ij} d^p(j, i)\right] \\
&\le \E[\cost_{y'}(j)] + \left(\left(\frac{d_{\max}(j)}{d_{\max}(j) - 2d(j, j')}\right)^p - 1\right) \cdot \E[\cost_{y'}(j)] \\
&\le \E[\cost_{y'}(j)] + \left(\left(\frac{1}{1 - 8(\varepsilon/p)^{4p}}\right)^p - 1\right) \cdot\E[\cost_{y'}(j)] \\
&=(1+O(\varepsilon))\cdot \E[\cost_{y'}(j)].
\end{align*}

This finishes the proof of the lemma.
\end{proof}

It remains to consider a type-3 client $j \in C$.  We can not guarantee a $1  + O_p(\varepsilon)$ loss for $j$; in the worst case, it may lose a factor of $2^p$ when converting $y'$ to $y''$. Let $j_1$ be its representative and $j_2$ be the nearest neighbor of $j_1$ in $C^* \setminus \{j_1\}$. Let $i_1, i_2$ be the unique facility in $B'_{j_1}, B'_{j_2}$ with positive $y'$ value respectively; they are also the unique facility in the two balls with positive $y''$ value. The following lemma says if we give an artificial upper bound $\frac{d(j, i_2)}2$ on the connection distance of $j$, then we can upper bound the expected cost of $j$ in the solution $y''$. In Section~\ref{sec:cost-new}, we show that the algorithm will open a facility with distance at $d(j, i_2)$ away with large enough probability. So, the weaker lemma suffices for our purpose.
\begin{lemma}
    \label{lemma:type-3-weaker-guarantee}
    \begin{align*}
        \E\Bigg[\sum_{i \in F}x''_{ij}\cdot\min\Big\{d^p(j, i), \left(\frac{d(j, i_2)}2\right)^p\Big\}\Bigg] \leq (1+O(p\varepsilon))\cost_{y'}(j).
    \end{align*}
\end{lemma}
\begin{proof}
    Similar to the proof of Lemma \ref{lem:type-2-cost}, we have

    \begin{align*}
    \mathrm{LHS}&\le \E\left[\left(\sum_{i\in B_{j_1}}y_i''d^p(j,i)\right)+(1-y''(B_{j_1}))_+\cdot \left(\frac{d(j, i_2)}2\right)^p\right] \\
    &\le \cost_{y'}(j)+(1-y'(B_{j_1}))\cdot \left(\frac{d(j, i_2)}2\right)^p+\sum_{i\in B_{j_1}}(y_i' - x'_{ij})d^p(j,i) -\sum_{i\notin B_{j_1}}x'_{ij}d^p(j,i).
    \end{align*}

    Denote $r=\frac{d(j_1,j_2)}{2}$ as the radius of $B_{j_1}$, $D_1=r-d(j,j_1)$ as a lower bound of $\min_{i\notin B_{j_1}}d(j,i)$, $D_2=\max\left(\frac{d(j, i_2)}2,r+d(j,j_1)\right)$ as an upper bound of $\max\left(\frac{d(j, i_2)}2,\max_{i\in B_{j_1}}d(j,i)\right)$. Since $\sum_{i\in B_{j_1}}(y'_i - x'_{ij}) + (1-y'(B_{j_1})) = \sum_{i \notin B_{j_1}}x'_{ij}$, we have

    \begin{align*}
        \mathrm{LHS} &\leq \cost_{y'}(j) + \left(\left(\frac{D_2}{D_1}\right)^p - 1\right)\sum_{i \notin B_{j_1}}x'_{ij}D_1^p \\
        &\le \cost_{y'}(j) + \left(\left(\frac{D_2}{D_1}\right)^p - 1\right)\sum_{i \notin B_{j_1}}x'_{ij}d^p(j,i) \\
        &\le \left(\frac{D_2}{D_1}\right)^p \cdot \cost_{y'}(j).
    \end{align*}

    It suffices to prove that $D_2/D_1=1+O(\varepsilon)$.

    Note that $B_{j_1}$ and $B_{j_2}$ are disjoint and each has $y$ value at least $\frac{1}{2}$. Thus, they can not be both inside $\mathrm{ball}^\circ_F(j,d_{\max}(j))$, which means $d(j,j_1)+d(j_1,j_2)+r\ge d_{\max}(j)$, as the radius of $B_{j_1}$ or $B_{j_2}$ is upper bounded by $r=\frac{d(j_1,j_2)}{2}$. Combined with $d(j,j_1)\leq 4d_\av(j)$ and $d_{\av}(j)<(\varepsilon/p)^{4p}d_{\max}(j)<\frac{1}{16}d_{\max}(j)$, we know that
    
    $$r\ge \frac{1}{4}d_{\max}(j)>\frac{1}{4(\varepsilon/p)^{4p}}d_{\av}(j)\ge\frac{1}{16(\varepsilon/p)^{4p}}d(j,j_1).$$

    Next we bound $d(j,i_2)$. Note that $d(j_2,i_2)\le \varepsilon d_{\max}(j_2)$ and $d_{\max}(j_2)\le d(j_1,j_2)+r=3r$, we have

    \begin{align*}
        d(j,i_2)&\le d(j,j_1)+d(j_1,j_2)+d(j_2,i_2)\\
        &\le 4d_{\av}(j)+2r+3\varepsilon r\\
        &\le (2+3\varepsilon+16(\varepsilon/p)^{4p})\cdot r.
    \end{align*}

    Thus,
    \begin{align*}
        \frac{D_2}{D_1}&\le\max\left(\frac{1+16(\varepsilon/p)^{4p}}{1-16(\varepsilon/p)^{4p}},\frac{1+3\varepsilon/2+8(\varepsilon/p)^{4p}}{1-16(\varepsilon/p)^{4p}}\right)
        =1+O(\varepsilon). \qedhere
    \end{align*}
\end{proof}

The following helper lemma will be needed in the analysis:
\begin{lemma}
    \label{lemma:max-j-2-big}
    If $d_{\max}(j_2) \ge \varepsilon d_{\max}(j_1)$, then $y'(B'_{j_2}) = y(B'_{j_2})\ge 1-4\varepsilon^{p+1}$.
\end{lemma}
\begin{proof}
    By the filtering condition, we have

    $$d_{\av}(j_2)\le\frac{1}{\left(1/(\varepsilon/p)^{4p}\right)-2}\cdot d(j_1,j_2)\le 2(\varepsilon/p)^{4p}\cdot d(j_1,j_2).$$

    We also know that

    $$\varepsilon d_{\max}(j_2) \ge \varepsilon^2 d_{\max}(j_1) \ge \frac{\varepsilon^2}2 \cdot d(j_1,j_2) \ge \frac{\varepsilon^2}2 \cdot \frac{1}{2(\varepsilon/p)^{4p}} d_{\av}(j_2) = \frac{p^{4p}}{4\varepsilon^{4p-2}}\cdot d_{\av}(j_2).$$

    By Markov inequality, we have $y(B'_{j_2})\ge1-\frac{d_{\av}(j_2)}{\varepsilon d_{\max}(j_2)}\ge 1-\frac{4\varepsilon^{4p-2}}{p^{4p}}\ge 1-4\varepsilon^{p+1}$.
\end{proof}

%% file: convertion.tex
\section{Iterative rounding algorithm with $k + O_{\varepsilon, p}(1)$ open facilities}
\label{sec:rounding_alg}

In this section, we describe the iterative algorithm that always opens $k + O_{\varepsilon, p}(1)$ facilities with probability $1 - O(\varepsilon)$. We analyze the number of open facilities in this section, while deferring the analysis of the connection cost to Section~\ref{sec:cost-new}.

Let $(x'', y'')$ be the solution obtained from the preprocessing step.  We define the global constants as follows. Let $c_1 = 12p^2$ so the solution $(x'', y'')$ is $\varepsilon^{c_1}$-integral. Let $c_5 = c_1 + 1 = 12p^2+1, c_2 = 48p^2+3 > 2(c_1 + c_5), c_3 = 132p^2+9 = 2(c_2+c_5) + c_1 + 1$ and $c_4 = 2c_1 + c_2 + c_5 +1= 84p^2+5$.

Let $\Delta := 1/\varepsilon^{c_1}$. It would be convenient for us to make copies of facilities in $F$, such that every copy corresponds to $\frac1\Delta = \varepsilon^{c_1}$ fraction of the facility. Therefore, we define $F^*$ to be the set containing $\Delta y''_i$ copies of facility $i$ for every $i \in F$.  Another global parameter we use throughout this and the next section is $L := \varepsilon^{c_2}$.  With $F^*$ defined, we do not need to use $x''$ and $y''$ most of the time from now on.

As in Section~\ref{sec:LMP}, we define a neighborhood graph: 
\begin{definition}[Neighborhood Graph]
    Given a set $F'$ containing copies of facilities in $F$ with $|F'| \geq \Delta$, the neighborhood graph $G = (F', E)$ is defined using the following procedure. Let $E \gets \emptyset$ initially. For every $i \in F'$, we take the $\Delta$ nearest facilities of $i$ in $F'$ (including $i$ itself), and add edges from these facilities to $i$ to $E$.
\end{definition}

Therefore, for the neighbor graph $G$ of $F'$, every facility $i \in F'$ has $\deg_G^-(i) = \Delta$. %

\begin{algorithm}[h]
    \caption{Iterative Rounding Algorithm with $k + O_{\varepsilon, p}(1)$ Open Facilities}
    \label{alg:iter-k}
    $F' \gets F^*$, $F_{\force} \gets \emptyset$\;
    \For{$t \gets 1$ to $T$, for some $T= \Theta\left(\frac{\log(1/\varepsilon)}{\varepsilon^{c_1+c_2}}\right)$}
    {
        create neighborhood graph $G = (F', E)$ for $F'$, define $F^+, F^-$ and $F^0$ as in text\; \label{line:iter-k-F+-0}
        run either unbalanced-update$()$ or balanced-update$()$, each with probability $1/2$, to update $F'$ and $F_\force$ \label{line:iter-k-update}
    }
    for every $i \in F$ do: let $\bar y_i \gets \frac{1}{\Delta}(\text{number of copies of $i$ in $F'$})$\; \label{line:iter-k-bary-F'}
    for every $i \in F_{\force} \subseteq F$ do: let $\bar y_i \gets 1$\; \label{line:iter-k-bary-F-Force} 
    treating $\bar y$ as a fractional solution to a weighted $k$-center problem, and use the $3$-approximation to round it to an integral solution (more details described in text); \Return the integral solution. 
    \label{line:iter-k-O(1)-approx}
\end{algorithm}

The iterative rounding algorithm is given in Algorithm~\ref{alg:iter-k}. Each facility in $F'$ corresponds to $\frac1\Delta$ fraction of a facility, while $F_\force \subseteq F$ is the set of facilities which we force to open integrally.

For a given directed graph $G' = (F'', E')$ (which is not necessarily a neighborhood graph), and $i \in F''$,
we define the \emph{imbalance} of any $i \in F''$ as $\imb_{G'}(i) := \frac{\deg_{G'}^-(i) - \deg_{G'}^+(i)}{\Delta}$. Indeed, the definition is relevant only when $\deg_{G'}^-(i) = \Delta$, in which case we have $\imb_{G'}(i) = 1 - \frac{\deg_{G'}^+(i)}{\Delta}$.

In Line~\ref{line:iter-k-F+-0}, we define $F^+$, $F^0$ and $F^-$ to be the sets of facilities in $i \in F'$ with positive, zero, and negative $\imb_{G'}(i)$ respectively. Thus, $F^+, F^0$ and $F^-$ form a partition of $F'$.  Then in Line~\ref{line:iter-k-update}, we run either unbalanced-update (Algorithm~\ref{alg:update-unbalanced}) or balanced-update (Algorithm~\ref{alg:update-balanced}), each with probability 1/2. In both procedures, we shall choose a set $I$ of facilities; this is as opposed to the LMP algorithm, which only chooses one facility in each iteration. As the names suggest, unbalanced-update chooses facilities in $F^+ \cup F^-$, while balanced-update chooses facilities in $F^0$. %
We now describe the two procedures separately.

\subsection{Unbalanced update (Algorithm~\ref{alg:update-unbalanced})}

\begin{algorithm}
    \caption{unbalanced-update()}
    \label{alg:update-unbalanced}
    let $A := \frac1\Delta\sum_{i \in F^+} \imb_G(i) = \frac1\Delta\sum_{i \in F^-}  |\imb_G(i)|$\; \label{line:update-unbalanced-A}
    \If{$A = O(1 / \varepsilon^{c_3})$ for some large enough hidden constant in $O(\cdot)$ notation \label{line:update-unbalanced-check-A}} 
    {
        for every $i^\circ$ with at least one copy in $F^+ \cup F^-$ do: $F_\force \gets F_\force \cup \{i^\circ\}$\; \label{line:update-unbalanced-force}
        $F' \gets F' \setminus (F^+ \cup F^-)$\;
        \Return\;
    }
    $R \gets \{i \in F^-: |\imb_G(i)| \geq A\varepsilon^{c_3}\}$\;
    add $\{i \in F:R \text{ contains at least 1 copy of $i$}\}$ to $F_\force$\; \label{line:update-unbalanced-force-1}
    remove $R$ from $F'$ and $G$\; \label{line:update-unbalanced-remove}
    create a new $G' := (F' \biguplus F_\mathrm{fict}, E \biguplus E_\mathrm{fict})$ as follows. Start from $G' = G = (F', E)$, $F_\mathrm{fict} = \emptyset$ and $E_\mathrm{fict} = \emptyset$. For every facility $i \in F$ such that there is at least one copy of $i$ in $F'$ and the in-degree of these copies in $G$ is $a < \Delta$, we add $\Delta - a$ \emph{fictitious facilities} to $F_\mathrm{fict}$, and \emph{fictitious edges} to $E_\mathrm{fict}$ from each of these facilities to each copy of $i$\;
    \label{line:update-unbalanced-fict}
    $I \gets \emptyset$\;
    for every $i \in F^+$, do: add $i$ to $I$ with probability $2L/\Delta$\; \label{line:update-unbalanced-choose-F+}
    for every $i \in (F^- \setminus R)\cup F_{\mathrm{fict}}$, do: add $i$ to $I$ with probability $2L \cdot (1 + \varepsilon^{c_5})/\Delta$\; \label{line:update-unbalanced-choose-F-}
    remove all out-neighbors of $I$ in $G'$ from $F'$\; \label{line:remove-out-neighbors}
    for every $i^\circ \in F$ with at least one copy in $I$: add $\Delta$ copies of $i^\circ$ to $F'$ \label{line:update-unbalanced-add-new-copies}
\end{algorithm}

Now we discuss unbalanced-update, which is defined in Algorithm~\ref{alg:update-unbalanced}. %
We need to explain Line~\ref{line:update-unbalanced-fict}. In Line~\ref{line:update-unbalanced-remove}, we remove $R$ from $F'$ and $G$. So some facilities in $F'$ may have in-degree less than $\Delta$. To fix the issue, we add a set $F_{\mathrm{fict}}$ of facilities to $G$, each with no incoming edges and 1 out-going edge leading to a facility $i \in F'$ with $\deg_G^-(i) < \Delta$. We add the facilities and edges so that in resulting graph $G'$, every $i \in F'$ has $\deg_{G'}^-(i) = \Delta$. We do not need to define distances for fictitious facilities, as they will never be open.

 In Line~\ref{line:update-unbalanced-choose-F+} and \ref{line:update-unbalanced-choose-F-}, we add facilities in $F'$ to $I$, with facilities in $F^+$ and $F^-$ added  with different probabilities. Notice that $F^+$ and $F^-$ were defined in Line~\ref{line:iter-k-F+-0} of Algorithm~\ref{alg:iter-k}, and they do not change as $F'$ does.

\medskip

In the rest of this section, we give a concentration bound in the net increase in $|F'|$ due to Line~\ref{line:remove-out-neighbors} and \ref{line:update-unbalanced-add-new-copies}. Let $X = (X_i)_{i \in (F'\setminus F^0)\cup F_\mathrm{fict}}$ be the indicator vector for $I$, where $F'$ and $F_\mathrm{fict}$ are w.r.t to the moment after we run Line~\ref{line:update-unbalanced-fict}. Let $Z(X)$ denote $1/\Delta$ times the increment of $|F'|$ after we update $F'$ in Line~\ref{line:remove-out-neighbors} and Line~\ref{line:update-unbalanced-add-new-copies}.

\begin{lemma}\label{lem:Lipschitz}
    The function $Z(X)$ satisfies the bounded differences property with bounds $\kappa_i$:
    \begin{align*}
        \kappa_i = 
        \begin{cases} 
            \max\{|\imb_{G'}(i)|, 1\} & \text{if } i \in F' \setminus F^0, \\
            \deg_{G'}^+(i)/\Delta & \text{if } i \in F_\mathrm{fict}.
        \end{cases}
    \end{align*}  
\end{lemma}

\begin{proof}
    Adding one facility $i$ to $I$ can result in $\Delta$ increment in $|F'|$ due to Line~\ref{line:update-unbalanced-add-new-copies}, and leads to at most $\deg_{G'}^+(i)$ decrement in $|F'|$ due to Line~\ref{line:remove-out-neighbors}. For fictitious facility $i\in F_{\mathrm{fict}}$, it can lead to at most $\deg_{G'}^+(i)$ decrement due to Line~\ref{line:remove-out-neighbors}. The lemma follows by the definition of $\imb_{G'}(i)$. 
\end{proof}

We aim to ensure that the net increment from the randomized choices does not exceed $O_{\varepsilon, p}(1)$ with exponentially high probability.

\begin{lemma}\label{lem:probability-bound}
    For a sufficiently small $\varepsilon$, we have:
    \begin{align*}
        \Pr[Z(X)\ge 1/\varepsilon^{c_4}] \le \exp(-\Theta(1/\varepsilon^{})).
    \end{align*}
\end{lemma}

\begin{proof}
    To begin, we bound the expectation of $Z(X)$. Let $h(X)$ denote $1/\Delta$ times the number of facilities removed in Line~\ref{line:remove-out-neighbors} and $g(X)$ denote $1/\Delta$ times the number of facilities added in Line~\ref{line:update-unbalanced-add-new-copies}. By definition, we have $Z(X)=g(X)-h(X)$.

    Let $q_i$ denote the probability that facility $i$ is added to $I$, and let $q(S) = \sum_{i\in S} q_i$. For simplicity, for each removed facility $i \in R$, we define $q_i := 2L\cdot (1+\varepsilon^{c_5})/\Delta$ and $\imb_{G'}(i) := 1 - \frac{\deg_G^+(i)}{\Delta}$. This definition gives $A = \frac{1}{\Delta}\sum_{i \in F^+} \imb_{G'}(i) = \frac{1}{\Delta}\sum_{i \in F^-} |\imb_{G'}(i)|$ and $\deg_G^+(i) = \Delta(1+|\imb_{G'}(i)|)$.
    
    For notational convenience, we define $F'' := (F^+\cup F^-)\setminus R$. We have $\mathbb{E}[g(X)] = (1/\Delta)\cdot \sum_{i\in F''} q_i \Delta = \sum_{i\in F''}q_i$.
    
    A facility will be removed in Line~\ref{line:remove-out-neighbors} if at least one of its in-neighbors is added to $I$. We have
    \begin{align*}
        \mathbb{E}[h(X)] &= \sum_{i\in F'} (1/\Delta) \cdot \left(1-\prod_{j\in \delta_{G'}^-(i)}(1-q_j)\right) \\
        &\ge \sum_{i\in F'} (1/\Delta)\cdot  \left( q(\delta^-_{G'}(i)) - \frac{1}{2}(q(\delta_{G'}^-(i)))^2 \right) \\
        &\ge \sum_{i\in F'}(1/\Delta)\cdot  q\left(\delta_{G'}^-(i)\right) \left(1-\Theta(\varepsilon^{c_2})\right) \tag{$q\left(\delta_{G'}^-(i)\right) \le 2L\cdot (1+\varepsilon^{c_5})$}\\
        &= \left(1-\Theta(\varepsilon^{c_2})\right) \sum_{i\in F''\cup F_{\mathrm{fict}}} (\deg_{G'}^+(i)/\Delta) \cdot q_i.
    \end{align*}
    
    By definition, the expected net increment is bounded by 
    \begin{align}
        \quad& \mathbb{E}[Z(X)] = \mathbb{E}[g(X)]-\mathbb{E}[h(X)] \nonumber\\
        &= \sum_{i\in F''} q_i -\left(1-\Theta(\varepsilon^{c_2})\right) \sum_{i\in F''\cup F_{\mathrm{fict}}} (\deg^+_{G'}(i)/\Delta) \cdot q_i \nonumber\\
        &= \sum_{i\in F''}q _i - \sum_{i\in F''\cup F_{\mathrm{fict}}} (\deg_{G'}^+(i)/\Delta) \cdot q_i +\Theta(\varepsilon^{c_2})\sum_{i\in F''\cup F_{\mathrm{fict}}}(\deg_{G'}^+(i)/\Delta)\cdot q_i \nonumber \\
        &=\sum_{i\in F^+} \imb_{G'}(i)q_i - \sum_{i\in F^-\setminus R} |\imb_{G'}(i)|q_i \nonumber\\
        &\quad - \sum_{i\in F_\mathrm{fict}} (\deg_{G'}^+(i)/\Delta)\cdot q_i + \Theta(\varepsilon^{c_2})\sum_{i\in F''\cup F_{\mathrm{fict}}}(\deg_{G'}^+(i)/\Delta)\cdot q_i \nonumber\\
        &= \sum_{i\in F^+} \imb_{G'}(i)q_i - \sum_{i\in F^-} |\imb_{G'}(i)|q_i +\sum_{i\in R}|\imb_{G'}(i)|q_i \nonumber\\
        &\quad - \sum_{i\in F_\mathrm{fict}} (\deg_{G'}^+(i)/\Delta)\cdot q_i + \Theta(\varepsilon^{c_2})\sum_{i\in F''\cup F_{\mathrm{fict}}}(\deg_{G'}^+(i)/\Delta)\cdot q_i \nonumber\\
        &= -2\varepsilon^{c_5}A\cdot L + \sum_{i\in R}|\imb_{G'}(i)|q_i \nonumber\\
        &\quad - \sum_{i\in F_\mathrm{fict}} (\deg_{G'}^+(i)/\Delta)\cdot q_i + \Theta(\varepsilon^{c_2})\sum_{i\in F''\cup F_{\mathrm{fict}}}(\deg_{G'}^+(i)/\Delta)\cdot q_i \label{ex-ineq5}\\
        &\le -2\varepsilon^{c_5}A\cdot L + \Theta(\varepsilon^{c_2})\sum_{i\in F''\cup F_{\mathrm{fict}}}(\deg_{G'}^+(i)/\Delta)\cdot q_i \label{ex-ineq2}\\
        &= -2\varepsilon^{c_5}A\cdot L + \Theta(\varepsilon^{c_2})\left(\sum_{i\in F''} (1-\imb_{G'}(i))q_i + \sum_{i\in R} (1-\imb_{G'}(i))q_i \right) \label{ex-ineq3}\\
        &\le -2\varepsilon^{c_5}A\cdot L + \Theta(\varepsilon^{c_2})\sum_{i\in F^+\cup F^-} (1+|\imb_{G'}(i)|)q_i \nonumber\\
        &\le -2\varepsilon^{c_5}A\cdot L + \Theta(\varepsilon^{c_2})\sum_{i\in F^+\cup F^-} \left(1/\varepsilon^{c_1}|\imb_{G'}(i)| + |\imb_{G'}(i)|\right)q_i \label{ex-ineq4}\\
        &\le -2\varepsilon^{c_5}A\cdot L + \Theta(\varepsilon^{c_2-c_1})\sum_{i\in F^+\cup F^-} |\imb_{G'}(i)|q_i \nonumber\\
        &= -2\varepsilon^{c_5}A\cdot L + \Theta(\varepsilon^{c_2-c_1})A\cdot L \nonumber\\
        &= -2\varepsilon^{c_5}A\cdot L + o(\varepsilon^{c_5})A\cdot L \label{ex-ineq6}\\
        &= -\Theta(\varepsilon^{c_5}) A\cdot L. \nonumber
    \end{align}

    For \eqref{ex-ineq5}, substituting $q_i = 2L/\Delta$ for $i \in F^+$ and $q_i = 2L\cdot (1+\varepsilon^{c_5})/\Delta$ for $i \in F^-$ yields $\sum_{i\in F^+} \imb_{G'}(i)q_i - \sum_{i\in F^-} |\imb_{G'}(i)|q_i = 2A\cdot L - 2(1+\varepsilon^{c_5})A\cdot L = -2\varepsilon^{c_5}A \cdot L$. To see \eqref{ex-ineq2}, notice that fictitious facilities exactly restore the missing incoming edges for the out-neighbors of facilities in $R$. Therefore, we have $\sum_{i\in R}(\deg_{G}^+(i)/\Delta) \cdot q_i= \sum_{i\in F_{\mathrm{fict}}}(\deg_{G'}^+(i)/\Delta)\cdot q_i$, which implies $\sum_{i\in R}|\imb_{G'}(i)|q_i- \sum_{i\in F_\mathrm{fict}} (\deg_{G'}^+(i)/\Delta)\cdot q_i\le 0$. 
    For \eqref{ex-ineq3}, we rewrite the term by mapping $F_{\mathrm{fict}}$ back to $R$. Since $\sum_{i\in F_{\mathrm{fict}}} \deg_{G'}^+(i)=\sum_{i\in R} \deg_G^+(i) = \sum_{i\in R} \Delta(1-\imb_{G'}(i))$ and the probabilities $q_i$ are identically defined as $2L(1+\varepsilon^{c_5})/\Delta$ for all facilities in $R \cup F_{\mathrm{fict}}$,  we can replace $F_\mathrm{fict}$ by $R$.
    \eqref{ex-ineq4} used the fact that all facilities $i\in F^+ \cup F^-$ satisfy $|\imb_{G'}(i)| \ge \varepsilon^{c_1}$. Finally, \eqref{ex-ineq6} holds for $c_2 > c_1+c_5$.

    Lemma~\ref{lem:Lipschitz} establishes that the function $Z(X)$ satisfies the bounded differences property. Applying McDiarmid’s Inequality, we have

    \begin{align}
        \Pr\left[Z(X)\ge 1/\varepsilon^{c_4}\right] &=  \Pr\left[Z(X)-\E[Z(X)] \ge 1/\varepsilon^{c_4} -\E[Z(X)]\right] && \nonumber\\
        &\le\Pr[Z(X)-\E[Z(X)]\ge 1/\varepsilon^{c_4}+\Theta(\varepsilon^{c_5})A\cdot L]&& \nonumber \\
        &\le \exp\left(-\frac{2(1/\varepsilon^{c_4}+\Theta(\varepsilon^{c_5})A\cdot L)^2}{\sum_{i\in F''}\max\{|\imb_{G'}(i)|,1\}^2+\sum_{i\in F_\mathrm{fict}}\kappa_i^2}\right) && \nonumber \\
        &\le \exp\left(-\frac{2(1/\varepsilon^{c_4}+\Theta(\varepsilon^{c_5})A\cdot L)^2}{\sum_{i\in F''}\max\{|\imb_{G'}(i)|,1\}^2+\sum_{i\in R} 1+|\imb_{G'}(i)|}\right) && \label{inequ:fict-to-R} \\
        &=\exp\left(-\Theta(1/\varepsilon)\cdot \frac{\left(1/\varepsilon^{c_4+1} + \sum_{i \in F^-\cup F^+} \varepsilon^{c_1+c_2+c_5-1}|\imb_{G'}(i)|\right)^2}{\sum_{i\in F''} 1/\varepsilon^{3} \max\{|\imb_{G'}(i)|,1\}^2+\sum_{i\in R}  1/\varepsilon^3|\imb_{G'}(i)|}  \right).&& \label{inequ:bab-probability}
    \end{align}

    For \eqref{inequ:fict-to-R}, since each fictitious facility $i\in F_{\mathrm{fict}}$ satisfies $\deg^+_{G'}(i) \le \Delta$, which implies $\kappa_i^2 = (\deg_{G'}^+(i)/\Delta)^2 \le \deg_{G'}^+(i)/\Delta$. Thus, we have $\sum_{i\in F_\mathrm{fict}} \kappa_i^2 \le \sum _{i\in F_\mathrm{fict}} \deg^+_{G'}(i)/\Delta = \sum_{i\in R}\deg_G^+(i)/\Delta = \sum_{i\in R}1+|\imb_{G'}(i)|$. 
    
    Focus on the numerator of the second term in \eqref{inequ:bab-probability}, it is bounded below by

    \begin{align*}
        \sum_{i \in F^-\cup F^+} (1/\varepsilon^{c_4+1} + A\varepsilon^{c_2+c_5-1}) \times \varepsilon^{c_1+c_2+c_5-1}|\imb_{G'}(i)|.
    \end{align*}

    Next, we show the second term in \eqref{inequ:bab-probability} is not less than $1$ by comparing the $i$-th term with the corresponding term in the denominator.
    
    For $i\in F''$, there are two cases:
    \begin{itemize}

    \item $|\imb_{G'}(i)|\ge 1$: Since $|\imb_{G'}(i)|$ is less than $A\varepsilon^{c_3}$, we have 

    \begin{align*}
        (1/\varepsilon^{c_4+1} + A\varepsilon^{c_2+c_5-1}) \times \varepsilon^{c_1+c_2+c_5-1}|\imb_{G'}(i)| &\ge A\varepsilon^{2(c_2+c_5-1)+c_1}\times |\imb_{G'}(i)|\\
        &\ge A\varepsilon^{c_3-3}\times |\imb_{G'}(i)| \tag{$2(c_2+c_5)+c_1+1\le c_3$}\\
        &\ge 1/\varepsilon^3 \cdot |\imb_{G'}(i)|^2.
    \end{align*}

    \item $|\imb_{G'}(i)|< 1$: Since all facilities $i\in F''$ satisfy $|\imb_{G'}(i)|\ge \varepsilon^{c_1}$, we have 
    \begin{align*}
        (1/\varepsilon^{c_4+1} + A\varepsilon^{c_2+c_5-1}) \times \varepsilon^{c_1+c_2+c_5-1}|\imb_{G'}(i)|&\ge 1/\varepsilon^{c_4+1}\cdot \varepsilon^{c_1+c_2+c_5-1}\cdot \varepsilon^{c_1} \\
        &= 1/\varepsilon^{2+c_4-c_2-c_5-2c_1}\\
        & \ge 1/\varepsilon^{3}.  \tag{$2c_1+c_2+c_5 +1\le c_4$}
    \end{align*}
    \end{itemize}

    For $i\in R$, we have
    \begin{align*}
        (1/\varepsilon^{c_4+1} + A\varepsilon^{c_2+c_5-1}) \times \varepsilon^{c_1+c_2+c_5-1}|\imb_{G'}(i)|&\ge \varepsilon^{c_1+c_2+c_5-c_4-2}|\imb_{G'}(i)|\\
        &\ge 1/\varepsilon^3\cdot|\imb_{G'}(i)| \tag{$c_1+c_2+c_5 +1\le c_4$}
    \end{align*}
    Therefore, the second term in \eqref{inequ:bab-probability} is not less than 1. We obtain

    \begin{align*}
        \Pr\left[Z(X)\ge 1/\varepsilon^{c_4}\right] &\le  \exp\left(-\Theta(1/\varepsilon)\cdot \frac{\left(1/\varepsilon^{c_4+1} + \sum_{i \in F^-\cup F^+} \varepsilon^{c_1+c_2+c_5-1}|\imb_{G'}(i)|\right)^2}{\sum_{i\in F''} 1/\varepsilon^{3} \max\{|\imb_{G'}(i)|,1\}^2+\sum_{i\in R}  1/\varepsilon^3|\imb_{G'}(i)|}  \right)\le \exp(-\Theta(1/\varepsilon)).
    \end{align*}
 which proves the lemma.
\end{proof}

\subsection{Balanced update (Algorithm~\ref{alg:update-balanced})}

\begin{algorithm}
    \caption{balanced-update()}
    \label{alg:update-balanced}
    $I \gets \emptyset$\;
    \For{every $i \in F^0$, using a random order of $F^0$}
    {
        if $\delta^+_G(i)$ is disjoint from $\delta^+_G(i')$ for every $i' \in I$, then add $i$ to $I$ with probability $\frac{2(1+\varepsilon^{c_5})L}{\Delta}$
    }
    remove all out-neighbors of $I$ in $G$ from $F'$\;
    for every $i^\circ \in F$ with at least one copy in $I$: add $\Delta$ copies of $i^\circ$ to $F'$ \label{line:update-balanced-add-new-copies}
\end{algorithm}

Now we move to the balanced-update procedure (Algorithm~\ref{alg:update-balanced}), which is relatively simpler than the unbalanced-update procedure.  We say two facilities $i$ and $i'$ in $F^0$ conflict with each other, if $\delta^+_G(i) \cap \delta^+_G(i') \neq \emptyset$. So, the set $I$ of facilities chosen is conflict-free. Moreover, as every $i \in I$ has $\imb_G(i) = 0$, we have the following claim: 
\begin{claim}
    The procedure balanced-update does not change $|F'|$ . %
\end{claim}

The following lemma says that the probability that any facility $i\in F^0$ being included in $I$ is close to $2L\cdot (1+\varepsilon^{c_5})y_i$.

\begin{lemma}
    For any facility $i\in F^0$, $\Pr[i\in I] = 2(1- o(\varepsilon^{c_5}))L\cdot (1+\varepsilon^{c_5})\cdot\frac{1}{\Delta}$.
\end{lemma}

\begin{proof}
    Two facilities $i_1, i_2 \in F^0$ have no conflict if and only if their out-neighbors have an empty intersection, i.e., $\delta^+_G(i_1) \cap \delta^+_G(i_2) = \emptyset$. 
    Since $\imb_G(i) = 0$ for any $i \in F^0$, its out-degree is exactly $\deg^+_G(i) = \Delta$. Because every facility in $G$ has an in-degree of exactly $\Delta$, the number of facilities sharing at least one out-neighbor with $i$ is bounded by $\Delta \times \Delta = \Delta^2$.
    
    Focus on a facility $i\in F^0$, analyze its probability of being included in $I$. Let $N(i)$ denote the set of the facilities that has conflict with $i$ and $B_{v}$ denote the event that $v$ is processed before $i$ and is added to $I$. We obtain:
    \begin{align*}
        \Pr[i\in I] = \left(1-\Pr\left[\bigvee_{v\in N(i)}B_v\right]\right)\frac{2(1+\varepsilon^{c_5})L}{\Delta} \ge \left(1-\sum_{v\in N(i)}\Pr[B_v]\right)\frac{2(1+\varepsilon^{c_5})L}{\Delta}.
    \end{align*}
    By definition, $\Pr[B_v] \le \Pr[v \in I] \le 2L(1+\varepsilon^{c_5})/\Delta = \Theta(\varepsilon^{c_1+c_2})$. Based on the previous analysis, we have $|N(i)| \le \Delta^2 = 1/\varepsilon^{2c_1}$. Thus, we can conclude $\sum_{v \in N(i)} \Pr[B_v] \le \Theta(\varepsilon^{c_1+c_2}) / \varepsilon^{2c_1} = \Theta(\varepsilon^{c_2-c_1})$. Since $c_2 > 2c_1+2c_5$, this upper bound is strictly $o(\varepsilon^{c_5})$, which proves the lemma.
\end{proof}

\subsection{Handling unconnected clients using a $3$-approximation for weighted $k$-center}
Now, we go back to Algorithm~\ref{alg:iter-k}, the iterative rounding algorithm. We only run for loop for $T = \Theta\left(\frac{\log(1/\varepsilon)}{\varepsilon^{c_1+c_2}}\right)$ iterations, instead of running it until the solution become integral.  We define $\bar y$ in Line~\ref{line:iter-k-bary-F'} and \ref{line:iter-k-bary-F-Force} using $F'$ and $F_\force$: every $i \in F'$ corresponds to $\frac1\Delta$ fractional opening, and every $i \in F_\force$ is integrally open.  

We then describe Line~\ref{line:iter-k-O(1)-approx}. For every client $j \in C$, we define $\bar d_{\max}(j) \geq 0$ to be the minimum real such that $\bar y(\ball_F(j, \bar d_{\max}(j))) \geq 1$. Then, $\bar y$ can be viewed as a fractional solution to the weighted $k$-center problem, where every $j$ has an individual connection requirement $\bar d_{\max}(j)$. Then, we  can round $\bar y$ to an integral solution $\tilde y$ that satisfies the following properties. First, $|\tilde y|_1 \leq \lceil |\bar y|_1\rceil$. Second, if $\bar y_i = 1$ for any $i \in F$, then $\tilde y_i = 1$. Finally, the connection distance of any $j \in C$ in the solution $\tilde y$ is at most $3 \bar d_{\max}(j)$.  We return the solution $\tilde y$.

\subsection{Counting number of open facilities}
We analyze the number of open facilities given by Algorithm~\ref{alg:iter-k}. %
In each iteration, the algorithm calls either unbalanced-update procedure (Algorithm~\ref{alg:update-unbalanced}) or balanced-update procedure (Algorithm~\ref{alg:update-balanced}). In the balanced-update procedure, no facilities are forcibly opened, and the net increment of fractional facilities in $F'$ is exactly zero.

In the unbalanced-update procedure, the algorithm deterministically adds at most $O(1/\varepsilon^{2c_1+c_3})$ facilities to $F_{\force}$ in Line~\ref{line:update-unbalanced-force} or Line~\ref{line:update-unbalanced-force-1}. By lemma~\ref{lem:probability-bound}, the probability that the net increment incurred in Lines~\ref{line:remove-out-neighbors} and \ref{line:update-unbalanced-add-new-copies} exceeds $1/\varepsilon^{c_4}$ is exponentially small, i.e., $\Pr[Z(X) \ge 1/\varepsilon^{c_4}] \le \exp(-\Theta(1/\varepsilon))$. The algorithm runs for loop for $T = \Theta\big(\frac{\log(1/\varepsilon)}{\varepsilon^{c_1+c_2}}\big)$ iterations. By applying the union bound over all $T$ iterations, the probability that $Z(X) \ge 1/\varepsilon^{c_4}$ in any of the iterations is at most:
\begin{align*}
    T \cdot \exp(-\Theta(1/\varepsilon)) = \Theta\left(\frac{\log(1/\varepsilon)}{\varepsilon^{c_1+c_2}}\right) \cdot \exp(-\Theta(1/\varepsilon)) = \exp(-\Theta(1/\varepsilon)).
\end{align*}

Therefore, with probability at least $1 - \exp(-\Theta(1/\varepsilon)) \geq 1 - O(\varepsilon)$, every iteration opens at most $1/\varepsilon^{2c_1+c_3} + 1/\varepsilon^{c_4}$ extra facilities, which guarantees that the algorithm ultimately opens at most $k + \Theta\big((1/\varepsilon^{2c_1+c_3} + 1/\varepsilon^{c_4})\cdot\frac{\log(1/\varepsilon)}{\varepsilon^{c_1+c_2}}\big)$ facilities. %

%% file: cost-proof-new.tex
\section{Analysis of the connection cost}

\label{sec:cost-new}
In this section, we show that the iterative rounding algorithm given in Sections~\ref{sec:preprocessing} and \ref{sec:rounding_alg} gives $(\alpha+\varepsilon)$-approximation for the problem, where $\alpha$ is the parameter defined in Theorem~\ref{thm:LMP}. However, due to the clients in $C'$, our approximation is also lower bounded by $2^p$. For general metrics, we have $\alpha = \frac{3^p + 1}{2} \geq 2^p$.  However, for Euclidean $k$-means, we can set our $\alpha$ to be $\frac{11}{3}$, but $2^p = 4$. We only get a $(4+\varepsilon)$-approximation for Euclidean $k$-means. The main theorem we prove in this section is the following:

\begin{theorem}
    \label{thm:cost}
    Let $\alpha \geq 2^p$ be a constant satisfying the property of Theorem~\ref{thm:LMP}.  Then, for every $j \in C$, the expected connection cost of $j$ is at most $(\alpha + 2^{O(p)}\varepsilon) \sum_{i \in F}x_{ij} d^p(i, j)$, over the randomness in the preprocessing procedure in Section~\ref{sec:preprocessing} and the iterative rounding algorithm (Algorithm~\ref{alg:iter-k}). 
\end{theorem}

Till the beginning of Section~\ref{subsec:type-3}, we fix a type-1 or 2 client $j \in C$ and prove the theorem for such a $j$. We shall show how to handle type-3 clients in Section~\ref{subsec:type-3}.  Also, until the very end, we fix the solution $(x'', y'')$ obtained from the preprocessing procedure. The expectations and probabilities are conditioned on this $(x'', y'')$. 

We define $F^*_j \subseteq F^*$ to be the set of $\Delta$ closest facilities in $F^*$. So, $\sum_{i \in F}x''_{ij} d^p(i, j) = \frac1\Delta\sum_{i \in F^*_j}d^p(i, j)$. Abusing notations slightly, for every $i \in F^*_j$, we simply use $d_i$ for $d(j, i)$ and let $d''_{\max}(j) = \max_{i \in F^*_j} d_i$. (The notation suggests that $d''_{\max}(j)$ is defined w.r.t solution $(x'', y'')$, to avoid confusion with the $d_{\max}(j)$ defined in Section~\ref{sec:preprocessing}). So, $\sum_{i \in F}x''_{ij} d^p(i, j) = \frac1\Delta\sum_{i \in F^*_j}d_i^p$, and our goal is to bound the expected connection cost of $j$ by $(\alpha + 2^{O(p)}\varepsilon) \cdot  \frac1\Delta\sum_{i \in F^*_j}d_i^p$. The following simple claim will be useful:
\begin{claim} \label{lem:d-double-prime-bound}
    If $j$ is a type-3 client, or a representative, we always have $d''_{\max}(j) \leq O(1) \cdot d_{\max}(j)$. 
\end{claim}
\begin{proof}
    Consider the case $j$ is a type-3 client. Let $j_1$ be its representative and $j_2$ be the nearest neighbor of $j_1$ in $C^*$. Then, $j_1$ and $j_2$ are $O(1) \cdot d_{\max}(j)$ away from $j$. Both balls $B_{j_1}$ and $B_{j_2}$ have radius $O(1)\cdot d_{\max}(j)$. Moreover, the $y'(B_{j_1})  + y'(B_{j_2}) \geq 1$. Therefore, $y''$ contains at least 1 fractional open facility in $B_{j_1} \cup B_{j_2}$. Therefore, $d''_{\max}(j) \leq O(1) \cdot d_{\max}(j)$.
    
    Consider the case $j$ is a representative. If $B_j = F_j$, then the claim clearly holds. Otherwise, let $j'$ be its nearest neighbor in $C^*$. Again, the balls $B_j$ and $B_{j'}$ will show that $d''_{\max}(j) \leq O(1) \cdot d_{\max}(j)$. 
\end{proof}

\subsection{Setup for the inductive proof}
  As in Section~\ref{sec:LMP}, we define a potential function $f^\mathrm{new}$, that is slightly different from $f$:
\begin{align}
    \label{equ:approx_bound}
    f^{\mathrm{new}}(S,b)\coloneqq(1+{\varepsilon})\cdot\bk*{\frac{\alpha}\Delta\sum_{i\in S}d^p_i+\left({1+2\varepsilon^{c_5}}-{\frac{|S|}{\Delta}}\right)b^p}, \qquad \forall S \subseteq F^*_j, b \geq 0.
\end{align}
Other than using $1/\Delta$ to replace $y_i$'s, the main difference is that we have the two factors $1+\varepsilon$ and ${1+2\varepsilon^{c_5}}$. \medskip

As in Section~\ref{sec:LMP}, we define $\Phi_j$ to be the connection cost of $j$ at the end of the algorithm.
Throughout the algorithm, we let $S = F^*_j \cap F'$, and $b$ be the minimum of $3d''_{\max}(j)$ and the distance between $j$ and its closest integrally open facility so far, where we say an original facility $i^\circ \in F$ is integrally open if either $i^\circ \in F_\force$ or there are $\Delta$ copies of $i^\circ$ in $F'$.   We define $\Phi'_j$ to be the value of $f^{\mathrm{new}}(S,b)$ at the end the for loop of Algorithm~\ref{alg:iter-k}. Similarly, we shall upper bound $\E[\Phi_j]$ by upper bounding $\E[\Phi'_j]$ using the $f^{\mathrm{new}}$ function, and the upper bound $\E[\Phi_j] - \E[\Phi'_j]$.

\begin{lemma}
\label{lemma:potential-1}
    Let $t \in [0, T]$. Suppose at the end of the $t$-th iteration, the state of $j$ is $(S, b)$. Conditioned on this event, we have $\E[\Phi'_j] \leq f^\mathrm{new}(S, b)$.
\end{lemma}

The rest of the section is devoted to the proof of Lemma~\ref{lemma:potential-1}.  Clearly, when $t = T$, the lemma holds by our definition of $\Phi'_j$.  We assume the lemma holds for $t \leq T$ and we show that it holds for $t - 1$.

\subsection{Inductive proof: bounding the cost for one iteration}
So, now we focus on the iteration $t$ of the for loop of Algorithm~\ref{alg:iter-k}, we run either unbalanced-update or balanced-update.  First we show that, if we run unbalanced-update in iteration $t$, then handling $R$ in Line \ref{line:update-unbalanced-force-1} and \ref{line:update-unbalanced-remove} of Algorithm~\ref{alg:update-unbalanced} can only decrease the value of $f^\textrm{new}(S, b)$.  Focus on some $i \in R \cap S$. Adding the original facility of $i$ to $F_\force$ only decreases $b$. After this operation, $b$ becomes at most $d_i$. Then removing $i$ from $S$ increases $f^{\mathrm{new}}(S, b)$ by $(1+\varepsilon)\left(-\frac{\alpha d_i^p}{\Delta} + \frac{b^p}{\Delta}\right) \leq 0$. Therefore, it suffices to prove the lemma by assuming $(S, b)$ is the state after we run Line~\ref{line:update-unbalanced-remove} of Algorithm~\ref{alg:update-unbalanced}.

If we run balanced-update, we simply define $F_\mathrm{fict} = \emptyset$, $E_\mathrm{fict} = \emptyset$  and $G' = G$.  After unifying the notations, we do not need to distinguish between the two procedures any more. Both procedures choose a set $I$ of facilities from $F' \cup F_\mathrm{fict}$. If $S \cap I \neq \emptyset$, then we define
\begin{align*}
    i_{\min}:=\text{argmin}_{i\in S\cap I}d_i
\end{align*} to be the facility in $S\cap I$ that is closest to $j$, and let $z_i=\Pr[i_{\text{min}}=i]$. In this case, we simply connect $j$ to $i_{\min}$ (which is integrally open) as in Section~\ref{sec:LMP}, and we say $j$ is happily connected. Otherwise, we let $i_{\min} = \bot$, and we let $(S', b')$ be the new state at the end of iteration $t$.

Let $b'=\min\{b,\min_{i\in I}d(i,j)\}$ denote the new backup connection cost, where we assume $d(i, j) = \infty$ if $i \in F_\mathrm{fict}$.  Define $b_i,b_P$ in the same way as in \cref{sec:connection_cost}, that is,
\begin{align*}
        b_i:= d_i + \min_{I \in S \setminus T_i}d(i, I), \forall i \in S, \quad \text{and}\quad  b_P:=\min_{i \in P}b_i,  \forall P \subseteq S.
\end{align*}
Recall from \cref{sec:connection_cost} that $b'\le \min\{b,b_{S\setminus S'}\}$.

For any $S_0\subseteq S$, let $x_{S_0}=\Pr[i_{\text{min}}=\bot\land S' = S_0]$. Therefore, we have $\sum_{i \in S}z_i + \sum_{S' \subseteq S} x_{S'} = 1$.
\begin{align*}
    V := \frac{1}{|S|}\sum_{i \in S}d^p_i \quad \land \quad V':=\sum_{i\in S}z_i\cdot d^p_i.
\end{align*}

Given the above definitions, we can upper bound the expected connection cost of $j$ as
\begin{align}
    & \sum_{i\in S}z_i\cdot d^p_i + \sum_{S'}x_{S'} f^{\mathrm{new}}(S', \min\{b,b_{S\setminus S'}\}) \nonumber\\
      ={}& V'+ (1+{\varepsilon})\sum_{S'}x_{S'}\bk*{\frac\alpha\Delta\sum_{i\in S'}d_i^p+\left({1+2\varepsilon^{c_5}}-{\frac{|S'|}{\Delta}}\right)\min\{b^p,b^p_{S\setminus S'}\}} \tag{definition of $f^{\mathrm{new}}$} \nonumber \\ 
    \le{}&V'+(1+{\varepsilon})(1-z(S))\left({1+2\varepsilon^{c_5}}-{\frac{|S|}{\Delta}}\right)b^p \hfill \nonumber \\
    &+ (1+{\varepsilon})\sum_{S'}x_{S'}\bk*{\frac\alpha\Delta\sum_{i\in S'}d_i^p+\left({\frac{|S|}{\Delta}}-{\frac{|S'|}{\Delta}}\right)\min\{b^p,b^p_{S\setminus S'}\}}
    \tag{$\sum_{S'}x_{S'} = 1-z(S)$} \nonumber\\
    \leq{}&V'+(1+{\varepsilon})(1-z(S))\left({1+2\varepsilon^{c_5}}-{\frac{|S|}{\Delta}}\right)b^p \nonumber\\
    \label{equ:ind_bound_2}
    &+(1+\varepsilon)\sum_{i\in S}\bk*{\frac{\alpha d_i^p}\Delta \sum_{S'\ni i}x_{S'}+\frac{\min\{b^p,b^p_i\}}\Delta\sum_{S'\not\ni i}x_{S'}}
\end{align}

\paragraph{Bounding the first term $V'$ in  \eqref{equ:ind_bound_2}.} 
To bound \eqref{equ:ind_bound_2}, in addition to \cref{equ:single_i_removed}, we bound the values of $V'$:

By \cref{clm:rounding_alg_props}, we have
        \begin{align*}
            z_i = \Pr[i_{\text{min}}=i]&\le \Pr[i\in I]\le \frac{(1+\varepsilon^{c_5})L}{\Delta}.
        \end{align*}
        So, 
        \begin{align*}
            V'&=\sum_{i\in S}z_i\cdot d^p_i \leq  \frac{(1+\varepsilon^{c_5})L}{\Delta} \sum_{i \in S} d^p_i =  {(1+\varepsilon^{c_5})L\cdot \frac{|S|}{\Delta}}\cdot {V}.
        \end{align*}

\paragraph{Bounding the second term in  \eqref{equ:ind_bound_2}.} 
Since $(1-z(S))=\Pr[i_{\text{min}}=\bot]$, we have
    \begin{align*}
        \Pr[i_{\min}=\bot]&=\Pr[S\cap I=\emptyset] 
        \le 1-\sum_{i\in S}\Pr[i \in I]+\sum_{\{i,i'\}\subseteq S,i \neq i'}\Pr[i, i' \in I] \\
        &\le 1-\frac{L|S|}{\Delta}+\left(\frac{(1+\varepsilon^{c_5})L|S|}\Delta\right)^2 
        \le 1-(1-\varepsilon^{2c_5})L\cdot {\frac{|S|}{\Delta}}.
    \end{align*}
    The last inequality used that $L=\varepsilon^{c_2}<\varepsilon^{2c_1+2c_5}$ is sufficiently small.
    
Therefore,
\begin{align*}
    \text{(second term in \eqref{equ:ind_bound_2})}&\le {(1+\varepsilon)\cdot}\left(1-L(1-\varepsilon^{2c_5})\cdot {\frac{|S|}{\Delta}}\right)\cdot\left({{1+2\varepsilon^{c_5}}}-{\frac{|S|}{\Delta}}\right)b^p.
\end{align*}

\paragraph{Bounding the third term in \eqref{equ:ind_bound_2}} For every $i\in S$, we bound the term
\begin{align}
    \label{equ:single_i_removed}
    \frac{\alpha d_i^p}\Delta \sum_{S'\ni i}x_{S'} + \frac{\min\{b^p,b^p_i\}}\Delta\sum_{S'\not\ni i}x_{S'}.
\end{align}

For this, we need the following properties of the rounding algorithm of \cref{sec:rounding_alg}.

\begin{claim}
    \label{clm:rounding_alg_props}
    The following holds over the randomness of $I$:
    \begin{itemize}
        \item For every $i \in F' \cup F_\mathrm{fict}$, 
            ${\frac{L}{\Delta}}\leq \Pr[i\in I] \leq {\frac{(1+\varepsilon^{c_5})L}{\Delta}}$.
        \item For every two facilities $i,i' \in F' \cup F_\mathrm{fict}$,
            $\Pr[i,i'\in I]\le 2((1+\varepsilon^{c_5})L)^2\cdot {\frac{1}{\Delta^2}}$.
    \end{itemize}
\end{claim}

For $i\in F^*_j$, let $T_i\subseteq F^*_j$ denote the set of facilities in $F^*_j$ that can remove $i$. In particular, $i\in T_i$. Let $U_i$ be the set of \emph{all} facilities that can remove $i$. Note that $|U_i|=\Delta$, and $|S\cup U_i|=\Delta+|S|-|T_i|$.

\begin{itemize}
    \item We first bound $\sum_{S'\ni i}x_{S'}$, which by definition is equal to $\Pr[i\in S'\land i_{\text{min}}=\bot]$. This is the probability that no facility in $S$ is selected, and $i$ is not removed. So we have
    \begin{align*}
        \Pr[i\in S'\land i_{\text{min}}=\bot]&= \Pr[I \cap (S \cup U_i) = \emptyset]\\
        &\le 1-\sum_{i'\in S\cup U_i}\Pr[i'\in I]+\sum_{\{i',i''\}\subseteq S\cup U_i,i' \neq i''}\Pr[i',i''\in I] \\
        &\le 1 - \frac{L|S \cup U_i|}{\Delta}+  \left(\frac{(1+\varepsilon^{c_5})L|S \cup U_i|}{\Delta}\right)^2\tag{by \cref{clm:rounding_alg_props}}\\
        &\le 1-(1-\varepsilon^{2c_5})L\left(1+{\frac{|S|}{\Delta}}-{\frac{|T_i|}{\Delta}}\right). \tag{$L=\varepsilon^{c_2}<\varepsilon^{2c_1+2c_5}$ is sufficiently small}
    \end{align*}
    \item Next, we bound $\sum_{S'\not\ni i}x_{S'}$. Note that $\sum_{S'\not\ni i}x_{S'}=\Pr[i\notin S'\land i_{\text{min}}=\bot]$. This is the probability that no facility in $S$ is selected, and $i$ is removed.  By the first bullet of \cref{clm:rounding_alg_props}, we have
    \begin{align*}
        \Pr[i\notin S'\land i_{\text{min}}=\bot]&\leq \Pr[I \cap (U_i \setminus T_i) \neq \emptyset]\le \frac{(1+\varepsilon^{c_5})L|U_i \setminus T_i|}{\Delta} = (1+\varepsilon^{c_5})L\left(1-{\frac{|T_i|}{\Delta}}\right).
    \end{align*}
    \item Then, we can bound the second term of \eqref{equ:single_i_removed} as follows:
    \begin{align*}
        \frac{\min\{b^p,b^p_i\}}\Delta\sum_{S'\not\ni i}x_{S'} &\le  \frac{\min\{b^p,b^p_i\}}{\Delta}\cdot(1+\varepsilon^{c_5})L\left(1-{\frac{|T_i|}{\Delta}}\right) \\
        &\le  (1+\varepsilon^{c_5})L\cdot\left(1-{\frac{|S|}{\Delta}}\right){\frac{b^p}{\Delta}} + (1+\varepsilon^{c_5})L\cdot\left({\frac{|S|}{\Delta}}-{\frac{|T_i|}{\Delta}}\right)\cdot {\frac{b_i^p}{\Delta}}  \\
        &\le  (1+\varepsilon^{c_5})L\cdot\left(1-{\frac{|S|}{\Delta}}\right){\frac{b^p}{\Delta}} +{\frac{(1+\varepsilon^{c_5})L}{\Delta^2}}\sum_{I\in S-T_i} (d_i+d(i,I))^p 
    \end{align*}
    \item We can then bound \eqref{equ:single_i_removed} as
    \begin{align*}
        &\quad \left(1-(1-\varepsilon^{2c_5})L\left(1+{\frac{|S|}{\Delta}}-{\frac{|T_i|}{\Delta}}\right)\right)\frac{\alpha d_i^p}\Delta 
        + (1+\varepsilon^{c_5})L\cdot\left(1-{\frac{|S|}{\Delta}}\right){\frac{b^p}{\Delta}}\\
        &+{\frac{(1+\varepsilon^{c_5})L}{\Delta^2}}\sum_{I\in S-T_i} (d_i+d(i,I))^p \\
        &\le \left(1-(1-\varepsilon^{2c_5})L\left(1+{\frac{|S|}{\Delta}}\right)\right)\frac{\alpha d_i^p}\Delta + (1+\varepsilon^{c_5})L\cdot\left(1-{\frac{|S|}{\Delta}}\right){\frac{b^p}{\Delta}} \\
     &\quad\quad+ \frac{(1+\varepsilon^{c_5})L}{\Delta^2}\sum_{I\in S}\max\{\alpha d^p_i,(d_i+d(i,I))^p\}.
    \end{align*}
    The inequality is obtained by moving an amount of $ \frac{(1 - \varepsilon^{2c_5})L|T_i|\alpha d_i^p}{\Delta^2}$ from the first term to the third term, and relax $1 - \varepsilon^{2c_5}$ to $ 1 + \varepsilon^{c_5}$.
\end{itemize}

So, the third term of \eqref{equ:ind_bound_2} can be bounded by taking sum of the bound for \eqref{equ:single_i_removed} over all $i \in S$:
\begin{align*}
	\text{(third term in \eqref{equ:ind_bound_2})} &\leq (1+\varepsilon)\sum_{i \in S}\Bigg[\left(1-(1-\varepsilon^{2c_5})L\left(1+{\frac{|S|}{\Delta}}\right)\right)\frac{\alpha d_i^p}\Delta + (1+\varepsilon^{c_5})L\cdot\left(1-{\frac{|S|}{\Delta}}\right){\frac{b^p}{\Delta}} \\
	     &\hspace*{70pt}+ \frac{(1+\varepsilon^{c_5})L}{\Delta^2}\sum_{I\in S}\max\{\alpha d^p_i,(d_i+d(i,I))^p\}\Bigg]\\
    &\leq(1+\varepsilon)\biggr[\left(1-(1-\varepsilon^{2c_5})L\left(1+{\frac{|S|}{\Delta}}\right)\right)\frac{\alpha|S|{V}}{\Delta}+(1+\varepsilon^{c_5})L\cdot\left(1-{\frac{|S|}{\Delta}}\right)\cdot {\frac{|S| b^p}{\Delta}} \\
    &\hspace*{70pt}+{\frac{(1+\varepsilon^{c_5})L|S|^2}{\Delta^2}}\cdot (2\alpha-1){V}\biggr].
\end{align*}
The inequality used the property of $\alpha$ in Theorem~\ref{thm:cost}, which is stated in Theorem~\ref{thm:LMP}.

\paragraph{Combining the bounds for all three terms in \eqref{equ:ind_bound_2}}

We need to show $\eqref{equ:ind_bound_2} \leq f^{\mathrm{new}}(S,b)$, which is equal to $(1+\varepsilon)\alpha \cdot {\frac{|S|}{\Delta}}{V}+(1+\varepsilon)\cdot ({1+2\varepsilon^{c_5}}-{\frac{|S|}{\Delta}})b^p$. It suffices to compare the coefficients for $\frac{|S|V}{\Delta}$ and $b^p$ separately, using the bounds for the three terms in \eqref{equ:ind_bound_2}. That is, we need to prove:
\begin{align}
    \label{ineq:happily_conn}
    (1+\varepsilon^{c_5})L + (1+\varepsilon)\alpha\left[\left(1 -(1-\varepsilon^{2c_5})L\left(1+{\frac{|S|}{\Delta}}\right)\right)+ (2\alpha-1)(1+\varepsilon^{c_5})L\cdot {\frac{|S|}{\Delta}} \right]\le (1+\varepsilon)\alpha
\end{align}
and
\begin{align}
    \label{ineq:backup_conn}
\left(1-L(1-\varepsilon^{2c_5})\cdot {\frac{|S|}{\Delta}}\right)\cdot\left({{1+2\varepsilon^{c_5}}}-{\frac{|S|}{\Delta}}\right) + (1+\varepsilon^{c_5})L \cdot\left(1-{\frac{|S|}{\Delta}}\right){\frac{|S|}{\Delta}}\le {1+2\varepsilon^{c_5}}-{\frac{|S|}{\Delta}}.
\end{align}

In the first inequality \eqref{ineq:happily_conn}, since $\varepsilon$ is sufficiently small compared to $\alpha$ (as assumed in the theorem statement) and $\alpha>1$, the LHS is increasing in ${\frac{|S|}{\Delta}}$, so we can assume that ${\frac{|S|}{\Delta}}=1$. In this case, the LHS is equal to
\begin{align*}
&\quad(1+\varepsilon^{c_5})L + (1+\varepsilon)\alpha\left[\left(1 -2(1-\varepsilon^{2c_5})L\right)+ (2\alpha-1)(1+\varepsilon^{c_5})L \right]\\
    =& \alpha(1+\varepsilon)+L\cdot \bk*{(1+\varepsilon^{c_5})-(1+\varepsilon)2\alpha(1-\varepsilon^{2c_5})+(1+\varepsilon)(2\alpha-1)(1+\varepsilon^{c_5})} \\
    =& \alpha(1+\varepsilon)+L\cdot \bk*{(1+\varepsilon)2\alpha(\varepsilon^{c_5}+\varepsilon^{2c_5})-\varepsilon(1+\varepsilon^{c_5})},
\end{align*}
where the latter term is at most $0$ when $\varepsilon$ is sufficiently small depending on $\alpha$.

The second inequality \eqref{ineq:backup_conn} trivially holds when $|S|=0$. When $|S|\ne 0$, \eqref{ineq:backup_conn} can be rewritten as

$$
\left(1+\varepsilon^{c_5}\right)\left(1-\frac{|S|}{\Delta}\right)\le\left(1-\varepsilon^{2c_5}\right)\left(1+2\varepsilon^{c_5}-\frac{|S|}{\Delta}\right),
$$

which holds since $\varepsilon^{c_5}$ is sufficiently small compared to $|S|/\Delta$. So, we completed the proof of Lemma~\ref{lemma:potential-1}.

\subsection{Wrapping up the analysis for type-1 and type-2 clients}

\begin{lemma} \label{lem:backup_distance}
The algorithm always opens a facility $i^*$ such that $d(j,i^*) \le 3d''_{\max}(j)$.
\end{lemma}
\begin{proof}
By the triangle inequality, the distance between any two facilities in $F^*_j$ is at most $2d''_{\max}(j)$. Since $|F^*_j|=\Delta$, any facility $i \in F^*_j$ must receive incoming edges from facilities within a distance of at most $2d''_{\max}(j)$ to accumulate a total fractional weight of 1.

Consider the first time if any facility $i \in F^*_j$ was removed from $F'$. This happens because one of its in-neighbors $i^*$ is integrally open, implying $d(i,i^*) \le 2d''_{\max}(j)$. By the triangle inequality, $d(j,i^*) \le d(j,i) + d(i,i^*) \le 3d''_{\max}(j)$.

When no facilities in $F^*_j$ was removed from $F'$ during the $T$ iterations, then maximum connection distance in the fractional solution $\bar y$ for the weighted $k$-center problem (defined in Line~\ref{line:iter-k-O(1)-approx} of Algorithm~\ref{alg:iter-k}) is at most $d''_{\max}(j)$. Then as we obtain a $3$-approximation in Line~\ref{line:iter-k-O(1)-approx}, some facility within distance at most $3d''_{\max}(j)$ to $j$ will be open.
\end{proof}

We now finish the proof of Theorem~\ref{thm:cost} for type-1 and type-2 clients $j$.  By Lemma \ref{lemma:potential-1}, since $|F^*_j| = \Delta$, we have 
\begin{align*}
\E[\Phi'_j] \leq f^\mathrm{new}(F^*_j, b) = (1+\varepsilon)\left(\alpha \sum_{i \in F^*_j} {\frac{1}{\Delta}} d^p_i + ({1+2\varepsilon^{c_5}} - 1)b^p\right) \leq (1+\varepsilon)(\alpha + O(\varepsilon)) \cdot \frac{1}{\Delta}\sum_{i \in F^*_j}d_i^p.
\end{align*}
The second inequality holds as we upper bounded $b$ by $3d''_{\max}(j)$ in its definition, and ${\frac{1}{\Delta}} \ge \varepsilon^{c_1}$ and $c_5 \ge c_1 + 1$. So, $O(\varepsilon^{c_5})d''^p_{\max}(j) \le O(\varepsilon) \varepsilon^{c_1}d''^p_{\max}(j)\le O(\varepsilon){\frac{1}{\Delta}} \sum_{i \in F^*_j}  d^p_i$.

Finally, we need to bound $\E[\Phi_j] - \E[\Phi'_j]$. First, upper bounding $b$ by $3d''_{\max}(j)$ is not an issue by Lemma~\ref{lem:backup_distance}. %
The actual cost $\Phi_j$ is bigger than $\Phi'_j$ only when after $T$ iterations, we have $F^*_j \cap F' \neq \emptyset$. (We assume the facilities added to $F'$ in Line~\ref{line:update-unbalanced-add-new-copies} of unbalanced-update and Line~\ref{line:update-balanced-add-new-copies} of balanced update are new copies, which are disjoint from $F^* \supseteq F^*_j$.) The probability that each facility $i \in F^*_j$ is removed from $F'$ is at least $\Delta\cdot \Omega(\frac L\Delta) = \Omega(\varepsilon^{c_2})$. Applying union bound, after $T = \Theta\left(\frac{\log(1/\varepsilon)}{\varepsilon^{c_1+c_2}}\right)$ iterations, the probability that $F^*_j \cap F' \neq \emptyset$ is bounded by $\Delta \left(1 - \Omega(\varepsilon^{c_2})\right)^{T} \le O(\varepsilon^{2c_1})$ if the hidden constant in $T$ is large enough. By Lemma~\ref{lem:backup_distance},  
\begin{align*}
  \E[\Phi_j] \leq \E[\Phi'_j] + O(\varepsilon^{2c_1} \cdot d''^p_{\max}(j)) = \E[\Phi'_j]+O(\varepsilon)\cdot \frac{1}{\Delta}\sum_{i \in F^*_j}d_i^p.  
\end{align*}
Applying Lemmas~\ref{lemma:type-1-cost} and \ref{lem:type-2-cost}, and deconditioning on $(x'', y'')$ proves Theorem~\ref{thm:cost} for type-1 and type-2 clients $j$.

\subsection{Handling type-3 clients}
\label{subsec:type-3}

Now we prove Theorem~\ref{thm:cost} for a type-3 client $j$. Recall that $C^*$ is the set of representatives we chose in the preprocessing step. Let $j_1$ be the representative of $j$, and $j_2$ be the nearest neighbor of $j_1$ in $C^* \setminus \{j_1\}$.  Let $i_1$ and $i_2$ be the unique facility in $B'_{j_1}$ and $B'_{j_2}$ with positive $y'$ ($y''$) value respectively. 

A main difference in the analysis is in the definition of the state $(S, b)$ for $j$. First, we let the backup distance $b$ be upper bounded by $\min\{3d''_{\max}(j), d(j, i_2)\}$, not just $3d''_{\max}(j)$. That is, at any time, $b$ is the minimum of $3d''_{\max}(j)$, $d(j, i_2)$ and the distance between $j$ and the closest integrally open facility. Because we have a backup distance $d(j, i_2)$, and $\alpha \geq 2^p$, we only include in $S$ the set of alive facilities $i \in F_j$ with $d(j, i) \leq d(j, i_2)/2$, since excluding them from $S$ can only decrease the potential function $f^\mathrm{new}(S, b)$.  So the initial $S$ may have $|S| < \Delta$; but this will not affect the analysis. Throughout the analysis, we shall explicitly state the conditions on the conditional expectations and probabilities. 

After running the iterative algorithm for $T$ iterations, we define $\Phi'_j$ to be the $f^\mathrm{new}(S, b)$ value at that moment. We have that $\E[\Phi'_j] \leq f^{\mathrm{new}}(S, b)$, where $(S, b)$ is the initial state for $j$. So,
\begin{align*}
    \E[\Phi'_j|(x'', y'')] &\leq f^{\mathrm{new}}(S, b) \\
    &\leq (1 + \varepsilon)\left(\alpha\sum_{i \in F: d(i, j) \leq d(j, i_2)/2} x''_{ij} d^p(i, j) + \sum_{i \in F: d(i, j) > d(j, i_2)/2} x''_{ij} d^p(j, i_2) + 2\varepsilon^{c_5} \cdot (3d''_{\max}(j))^p \right) \\
    &\leq (\alpha + O(\varepsilon))\sum_{i \in F} x''_{ij} \min\left\{d^p(i, j), \left(\frac{d(j, i_2)}{2}\right)^p\right\} + O(\varepsilon^{c_5}) \cdot 2^{O(p)} \cdot d''^p_{\max}(j)).
\end{align*}
The second inequality used that $\alpha \geq 2^p$.

So, decondition on $(x'', y'')$ and $(x', y')$, we obtain 
\begin{align*}
    \E[\Phi'_j] \leq (1 + O(p\varepsilon))\alpha \sum_{i \in F} x_{ij} d^p(i, j) + O(\varepsilon^{c_5}) \cdot 2^{O(p)} \cdot \E[d''^p_{\max}(j)].
\end{align*}
The inequality used Lemmas~\ref{lemma:y-to-y'} and \ref{lemma:type-3-weaker-guarantee}. 

The pipage rounding procedure can guarantee that $\E[d''^p_{\max}(j) | (x', y')] \leq \Delta \cdot 2^{O(p)} \sum_{i \in F} x'_{ij} d^p(i, j)$. Deconditioning gives us $\E[d''^p_{\max}(j)] \leq \Delta \cdot 2^{O(p)} \sum_{i \in F} x_{ij} d^p(i, j)$. As $\Delta \cdot \varepsilon^{c_5} < \varepsilon$, we have
\begin{align}
    \label{inequ:Phi'-j-type-3}
    \E[\Phi'_j] \leq \left(1 + 2^{O(p)}\cdot \varepsilon\right)\cdot \alpha \sum_{i \in F} x_{ij} d^p(i, j).
\end{align}

It remains to bound $\E[\Phi_j - \Phi'_j]$. A crucial lemma we need is the following: %
\begin{lemma}
    \label{lemma:two-small-balls}
    Let $i_1, i_2 \in F$ be two different facilities with $y''_{i_1} > 0.9, y''_{i_2} > 0.9$ and $D: = d(i_1, i_2)$. Then, with probability at least $1 - (1+\varepsilon^{c_5})^2(1 - y''_{i_1})(1 - y''_{i_2})$, some facility in $\ball_F(i_1, D)$ is integrally open in Algorithm~\ref{alg:iter-k}.
\end{lemma}

\begin{proof}%
    If some facility in $\ball_F(i_1, D)$ is integrally open, we say the good event happens; otherwise we say the bad event happens. Notice that copies of $i_1$ ($i_2$ resp.) will behave the same during the course of Algorithm~\ref{alg:iter-k}, as they always have the same set of in-neighbors. Therefore, for convenience we assume $i_1$ ($i_2$ resp.) is one copy of itself in $F^*$, and then we can use $i_1$ ($i_2$ resp.) to represent all its copies.

    For the bad event to happen, we can not forcibly open $i_1$ or $i_2$. Moreover, we should remove $i_2$ from $F'$ in some iteration, and then remove $i_1$ from $F'$ in a later iteration. This holds as if both $i_1$ and $i_2$ are alive in $F'$, then the in-neighbors of $i_1$ will have distance at most $D$ to $i_1$ at the beginning of an iteration of Algorithm~\ref{alg:iter-k}. If $i_1$ is removed, then the good event happens (this also covers the case where some in-neighbor of $i_1$ was forcibly open during Algorithm~\ref{alg:update-unbalanced}).

    We break the bad event into two sub-events. Event 1 happens when we remove $i_2$ in an iteration, but the good event does not happen. This implies that we did forcibly open or choose any in-neighbors of $i_2$, but we have choose some in-neighbor of $i_2$ that is not a copy of $i_2$.  So, event 1 happens with probability at most $(1+\varepsilon^{c_5})\cdot \frac{1 - y''_{i_2}}{2 - y''_{i_2}} \leq (1+\varepsilon^{c_5})\cdot (1-y''_{i_2})$.

    Now we condition on that event 1 happens. Event 2 happens when we remove $i_1$ without choosing its copies. This happens with probability at most $(1 + \varepsilon^{c_5})\cdot (1 - y''_{i_1})$. Then the lemma follows.
\end{proof}

\begin{lemma}
    $\E[\Phi_j | (x'', y'')] \leq (1+O(p\varepsilon))\E[\Phi'_j|(x'', y'')] + 2^{O(p)} \cdot \varepsilon \sum_{i \in F}x''_{ij}d^p(i, j)$. 
\end{lemma}

\begin{proof}

First, if $d_{\max}(j_2) < \varepsilon d_{\max}(j_1)$, then we always guarantee that there is an integrally open facility within $3d''_{\max}(j_2) \leq O(1) \cdot d_{\max}(j_2)$ distance away from $j_2$. The distance between this facility and $j$ is at most
\begin{align*}
    d(j, j_2) + O(1)\cdot d_{\max}(j_2) \leq (1 + O(\varepsilon))d(j, i_2) + O(\varepsilon) d_{\max}(j_1) \leq (1 + O(\varepsilon)) d(j, i_2).
\end{align*}
Therefore, we have $\E[\Phi_j | (x'', y'')] \leq (1+O(p\varepsilon))\E[\Phi'_j|(x'', y'')]$ in this case.

Now, assume $d_{\max}(j_2) \geq \varepsilon d_{\max}(j_1)$. By Lemma~\ref{lemma:max-j-2-big}, we have $y(B'_{j_2}) \geq 1 - 4\varepsilon^{p+1}$.  So, the unique facility $i_2$ in $B'_{j_2}$ with positive $y''$ value has $y''_{i_2} \geq 1 - 4\varepsilon^{p+1}$. Let $i_1$ be the unique facility in $B'_{j_1}$ with positive $y''$ value.  Notice that we have $y''_{i_1} \geq 1 - \frac{\sum_{i \in F}x''_{ij_1}d^p(i, j_1)}{(\varepsilon d_{\max}(j_1))^p}$. 

Since $j_1$ is the representative of $j$, we have $\sum_{i \in F}x''_{ij_1}d^p(i, j_1) \leq \sum_{i \in F}x''_{ij}d^p(i, j)$. By Lemma~\ref{lem:not-type1-distance}, we have $d_{\max}(j_1) \geq d_{\max}(j) - d(j, j_1) \geq d_{\max}(j) - 4d_{\av}(j)$. Moreover, client $j$ is a type-3 client, which implies $4d_{\av}(j) < 4(\varepsilon/p)^{4p} d_{\max}(j) \leq \varepsilon d_{\max}(j)$. Thus, we obtain 

$$y''_{i_1} \geq 1 - \frac{\sum_{i \in F}x''_{ij}d^p(i, j)}{(\varepsilon(1-\varepsilon) d_{\max}(j))^p} = 1 - \frac{(1+O(p\varepsilon))\sum_{i \in F}x''_{ij}d^p(i, j)}{(\varepsilon d_{\max}(j))^p}.$$

We apply Lemma~\ref{lemma:two-small-balls} with our $i_1$ and $i_2$. The probability that there is integrally open facility with at most $d(i_1, i_2)$ distance away from $i_1$ is at least $1 - (1+\varepsilon^{c_5})^2\cdot \frac{(1+O(p\varepsilon))\sum_{i \in F}x''_{ij}d^p(i, j)}{(\varepsilon d_{\max}(j))^p}\cdot 4\varepsilon^{p+1} = 1 - \frac{O(\varepsilon^{}) \sum_{i \in F}x''_{ij}d^p(i, j)}{(d_{\max}(j))^p}$. As Lemma~\ref{lem:backup_distance} says that there is always an open facility with distance $3d''_{\max}(j)$ away from $j$, which is bounded by $O(1)d_{\max}(j)$ according to Claim~\ref{lem:d-double-prime-bound} , we have 
\begin{align*}
    \E[\Phi_j|(x'', y'')] &\leq \E[\Phi'_j|(x'', y'')] + \frac{O(\varepsilon) \sum_{i \in F}x''_{ij}d^p(i, j)}{(d_{\max}(j))^p} \cdot 2^{O(p)} d^p_{\max}(j) \\
    &= \E[\Phi'_j|(x'', y'')]+ 2^{O(p)} \cdot \varepsilon \sum_{i \in F}x''_{ij}d^p(i, j). \qedhere
\end{align*}
\end{proof}

Finally, deconditioning on $(x', y')$ gives us 
\begin{align*}
\E[\Phi_j] \leq (1+O(p\varepsilon))\E[\Phi'_j] + 2^{O(p)} \cdot \varepsilon \sum_{i \in F}x_{ij}d^p(i, j).    
\end{align*}
Notice that even though for type-3 clients $j$, $\E[\sum_{i \in F}x''_{ij}d^p(i ,j)]$ may not be bounded by $(1+O(p\varepsilon))\sum_{i \in F}x_{ij}d^p(i ,j)$, it is bounded by $2^{O(p)} \sum_{i \in F}x_{ij}d^p(i ,j)$. Combine the above inequality with \eqref{inequ:Phi'-j-type-3} gives us 
\begin{align*}
    \E[\Phi_j]  \leq (1+2^{O(p)}\varepsilon)\cdot \alpha \sum_{i \in F}x_{ij}d^p(i, j).
\end{align*}

This finishes the proof of Theorem~\ref{thm:cost} for type-3 clients.

%% file: pseudo-solution-to-solution.tex
\appendix

\section{Obtaining solutions from additive pseudo-solutions} \label{sec:pseudo_approximation_reduction}

\begin{theorem} \label{thm:pseudo-to-sol}
    For any constant $p\ge 1$, suppose that $A$ is a $c$-additive $\alpha$-approximation algorithm for $k$-clustering with cost function being the $p$-th power of the distance. Then for any $\varepsilon>0$, there exists a $(\alpha+\varepsilon)$-approximation algorithm $A'$ for $k$-clustering with the same cost function, whose running time is $n^{O((\gamma p)^p\cdot \alpha c/\varepsilon)}$ times that of $A$, where $\gamma$ is a global constant.
\end{theorem}

Theorem~\ref{thm:pseudo-to-sol} is a generalization of Theorem 4 in~\cite{li2013approximating}. The proofs are very similar.

\begin{definition}
    For a $k$-clustering instance $\cI$, denote $\mathrm{opt}_{\cI}$ as the minimal cost and $\mathrm{OPT}_{\cI}$ as the set of selected facilities that minimizes the cost. Denote $\mathrm{CBall}_{\cI}(i,r)$ ($\mathrm{FBall}_{\cI}(i,r)$ resp.) as the set of clients (facilities, resp.) with distance \textbf{strictly} less than $r$ from (facility or client) $i$.

    For $A>0$, a facility $i\in F$ is said to be $A$-dense if

    $$\Big((1-\xi)d(i,\mathrm{OPT}_{\cI})\Big)^p\cdot |\mathrm{CBall}_{\cI}(i,\xi d(i,\mathrm{OPT}_{\cI})|>A,$$

    where $\xi =1/3$. $i$ is $A$-sparse otherwise.

    The instance $\cI$ is said to be $A$-sparse if every facility is $A$-sparse.
\end{definition}

To obtain an $(\alpha+\varepsilon)$-approximation solution, we first process the original instance $\cI$ into a $\mathrm{opt}_{\cI}/t$-sparse instance $\cI'$, such that an optimal solution $\mathrm{OPT}_{\cI}$ is also an optimal solution of $\cI'$. Then we only need to find an $(\alpha+\varepsilon)$-approximation solution in the sparse instance.

\subsection{Reduction to a sparse instance}

\begin{lemma} \label{lem:reduce-to-sparse}
    For a $k$-clustering instance $\cI$, there exists an algorithm that runs in time $n^{O(t)}$ and outputs $n^{O(t)}$ instances obtained by removing several facilities from $\cI$, such that at least one of such instances $\cI'$ satisfies:
    \begin{itemize}
        \item $\mathrm{OPT}_{\cI}$ is also an optimal solution to $\cI'$;
        \item $\cI'$ is $\mathrm{opt}_{\cI}/t$ sparse.
    \end{itemize}
\end{lemma}

\begin{algorithm}[h]
    \caption{Reduction to A Sparse Instance}
    \label{alg:reduce-to-sparse}
        \For{each $t'\leq t$ facility pairs $(i_1,i'_1),\ldots,(i_{t'},i'_{t'})$}
        {
            Let $F'=F\setminus \bigcup_{x=1}^{t'}\mathrm{FBall}_{\cI}(i_x,d(i_x,i'_x))$
            
            Output the instance obtained from $\cI$ by reducing the set of facilities to $F'$
        }
\end{algorithm}

\begin{proof}
    We will prove that Algorithm~\ref{alg:reduce-to-sparse} satisfies the description of Lemma~\ref{lem:reduce-to-sparse}. Clearly, Algorithm~\ref{alg:reduce-to-sparse} runs in time $n^{O(t)}$ and outputs $n^{O(t)}$ instances.

    Let $(i_1,i'_1),\ldots,(i_q,i'_q)$ be the longest sequence of facility pairs such that:
    \begin{itemize}
        \item $i_x\notin\mathrm{OPT}_{\cI}$, and $i'_x\in\mathrm{OPT}_{\cI}$ is the closest facility to $i_x$ in $\mathrm{OPT}_{\cI}$.
        \item $i_x$ is $\mathrm{opt}_{\cI}/t$-dense.
        \item $i_x\notin F\setminus \bigcup_{y=1}^{x-1}\mathrm{FBall}_{\cI}(i_y,d(i_y,i'_y))$.
    \end{itemize}

    Let $F'=F\setminus \bigcup_{x=1}^{q}\mathrm{FBall}_{\cI}(i_x,d(i_x,i'_x))$, it is easy to see that $\mathrm{OPT}_{\cI}\subseteq F'$ and that any facility in $F'$ is $\mathrm{opt}_{\cI}/t$-sparse. Thus it suffices to show that Algorithm~\ref{alg:reduce-to-sparse} enumerates $(i_1,i'_1),\ldots,(i_q,i'_q)$, i.e. $q\leq t$.

    Denote $\mathcal B_x$ as $\mathrm{CBall}_{\cI}(i_x,\xi d(i_x,i'_x))$. Since $i_x$ is $\mathrm{opt}_{\cI}/t$-dense, we know that

    $$\sum_{c\in\mathcal B_x}\text{connection cost of }c\ge |\mathcal B_x|\cdot \Big((1-\xi)d(i_x,i'_x)\Big)^p\ge\mathrm{opt}_{\cI}/t.$$

    Moreover, for any $x\neq y$, we know that

    $$\begin{aligned}
        \xi(d(i_x,i'_x)+d(i_y,i'_y))&\leq \xi(d(i_x,i'_x)+d(i_y,i'_x))\\
        &\leq \xi(2d(i_x,i'_x)+d(i_y,i_x))\\
        &\leq 3\xi d(i_y,i_x)\\
        &=d(i_y,i_x).
    \end{aligned}$$

    Thus $\mathcal B_x\cap\mathcal B_y=\varnothing$. This implies

    $$\mathrm{opt}_{\cI}\ge \sum_{x=1}^q\sum_{c\in\mathcal B_x}\text{connection cost of }c\ge q\cdot \mathrm{opt_{\cI}}/t.$$

    Hence $q\leq t$.
\end{proof}

\subsection{Solving the sparse instance}

\begin{lemma} \label{lem:solve-sparse}
    For an $A$-sparse instance $\cI$, given a $c$-additive pseudo solution $T$, $\delta\in(0,1/6)$ and $t\ge 2c\cdot\left(\frac{2}{\delta\xi}\right)^p$, there exists an algorithm that runs in time $n^{O(t)}$ and returns a solution $S$ such that

    $$\mathrm{cost}(S)\leq\max\left(\mathrm{cost}(T)+cB,\left(\frac{\xi+\delta}{\xi-\delta}\right)^p\cdot\mathrm{opt}_{\cI}\right),$$

    where $B=2\cdot\left(A+\frac{\mathrm{cost}(T)}{t}\right)\cdot\left(\frac{2}{\delta\xi}\right)^p$.
\end{lemma}

\begin{algorithm}[h]
    \caption{Solving the Sparse Instance}
    \label{alg:solve-sparse}
        $T'\leftarrow T$

        \While{$|T'|>k$ and $\exists i\in T'$ s.t. $\mathrm{cost}(T'\setminus\{i\})\leq \mathrm{cost}(T')+B$}
        {
            Remove $i$ from $T'$
        }

        \If{$|T'|=k$}
        {
            \Return{$T'$}
        }
    
        \For{each $(D,V)$ such that $D\subseteq T',V\subseteq F$, $|D|+|V|=k$ and $|V|<t$}
        {
            $\forall i\in D$, let $L_i=d(i,T'\setminus\{i\})$. Let $f_i$ be the facility in $\mathrm{FBall}_{\cI}(i,\delta L_i)$ that minimizes

            $$\sum_{j\in\mathrm{CBall}_{\cI}(i,\xi L_i)}\min(d^p(j,f_i),d^p(j,V))$$
            
            Let $S_{D,V}=V\cup\{f_i\mid i\in D\}$
        }

        \Return{$S_{D,V}$ with the minimal cost}
\end{algorithm}

\begin{proof}
    We will prove that Algorithm~\ref{alg:solve-sparse} satisfies the description of Lemma~\ref{lem:solve-sparse}. Clearly, Algorithm~\ref{alg:solve-sparse} runs in time $n^{O(t)}$.

    If the algorithm directly returns $T'$, it removes at most $c$ facilities and each removal contributes at most $B$ to the total cost, thus

    $$\mathrm{cost}(T')\leq \mathrm{cost}(T)+cB.$$

    If the algorithm returns $S_{D,V}$, it means that $|T'|>k$ and $\mathrm{cost}(T'\setminus\{i\})>\mathrm{cost}(T')+B$ for any $i\in T'$. It suffices to construct $D_0,V_0$ such that $S_{D_0,V_0}\leq \left(\frac{\xi+\delta}{\xi-\delta}\right)^p\cdot\mathrm{opt}_{\cI}$.

    \begin{definition}
        For $i\in T'$, let $L_i=d(i,T'\setminus\{i\})$ and $\ell_i=d(i,\mathrm{OPT}_{\cI})$. $i$ is called determined if $\ell_i<\delta L_i$, otherwise $i$ is undetermined.
    \end{definition}

    Let $D_0$ be the set of determined facilities. Denote $f^*_i$ as the closest facility in $\mathrm{OPT}_{\cI}$ to $i$. Let $V_0=\mathrm{OPT}_{\cI}\setminus\{f^*_i\mid i\in D_0\}$.

    \begin{claim}
        $|D_0|+|V_0|=k$, and $|V_0|<t$.
    \end{claim}

    \begin{proof}[Proof of Claim]
    We first prove that $|D_0|+|V_0|=k$. Since $|V_0|=|\mathrm{OPT}_{\cI}|-|\{f^*_i\mid i\in D_0\}|$, it suffices to prove that $f^*_i$ are pairwise distinct. Actually, suppose that $f^*_i=f^*_{i'}$ for $i,i'\in D_0$, then

    $$\max(L_i,L_{i'})\le d(i,i')\le d(i,f^*_i)+d(i',f^*_i)=l_i+l_{i'}<2\delta\max(L_i,L_{i'}),$$

    Leading to a contradiction since $\delta<1/2$ always holds.
    
    We then prove that $|V_0|<t$. Denote $U_0$ as the set of undetermined facilities in $T'$. We know that $|U_0|=|T'|-|D_0|>k-|D_0|=|V_0|$, thus it suffices to prove that $|U_0|\le t$.

    Suppose that $|U_0|>t$. Denote $C_i$ as the set of all clients that connects to facility $i$, and $Con_i=\sum_{c\in C_i}d^p(c,i)$ as the connection cost of $C_i$. Select $i$ as the facility in $U_0$ with minimal $Con_i$, then $Con_i<\mathrm{cost}(T')/t$. Let $i'$ be the closest facility in $T'\setminus\{i\}$ to $i$, we consider the connection cost of $C_i$ if they are connected to $i'$ instead, i.e. $\sum_{c\in C_i}d^p(c,i')$. We divide $C_i$ into two groups:
    
    \begin{itemize}
        \item $C_i\cap\mathrm{CBall}(i,\delta\xi L_i)$, denoted as $C_i^0$.
    
        Since $i$ is undetermined, we know that $\delta L_i\leq \ell_i$, so $\mathrm{CBall}(i,\delta\xi L_i)\subseteq \mathrm{CBall}(i,\xi \ell_i)$.
        
        Since $i$ is $A$-sparse, we know that
    
        $$\begin{aligned}
            \sum_{c\in C_i^0}d^p(c,i')&\leq \Big((1+\delta\xi)L_i\Big)^p\cdot |C_i^0|\\
            &\le \left(\frac{(1+\delta\xi)L_i}{(1-\xi)\ell_i}\right)^p\cdot \Big((1-\xi)\ell_i\Big)^p\cdot |\mathrm{CBall}(i,\xi \ell_i)|\\
            &\le \left(\frac{(1+\delta\xi)}{\delta(1-\xi)}\right)^p\cdot A\\
            &\le \left(\frac{2}{\delta\xi}\right)^p\cdot A.
        \end{aligned}$$
    
        \item $C_i\setminus \mathrm{CBall}(i,\delta\xi L_i)$, denoted as $C_i^1$.
    
        For any $c\in C_i^1$, we have that $d(c,i)\ge\delta\xi L_i$ and $d(c,i')\leq d(c,i)+L_i$, thus,
    
        $$\frac{d(c,i')}{d(c,i)}\leq 1+\frac{1}{\delta\xi}<\frac{2}{\delta\xi}.$$
    
        This implies
    
        $$\begin{aligned}
            \sum_{c\in C_i^1}d^p(c,i')&\leq \left(\frac{2}{\delta\xi}\right)^p\cdot \sum_{c\in C_i^1}d^p(c,i)\\
            &\le \left(\frac{2}{\delta\xi}\right)^p\cdot\frac{\text{cost}(T')}{t}.
        \end{aligned}$$
    \end{itemize}

    Thus the increase of connection cost is upper bounded by:

    $$\begin{aligned}
        \left(A+\frac{\mathrm{cost}(T')}{t}\right)\cdot\left(\frac{2}{\delta\xi}\right)^p&\le \left(A+\frac{\mathrm{cost}(T)}{t}+\frac{cB}{t}\right)\cdot\left(\frac{2}{\delta\xi}\right)^p\\
        &\le \frac{B}{2}+\left(\frac{cB}{t}\right)\cdot\left(\frac{2}{\delta\xi}\right)^p\\
        &\le B.
    \end{aligned}$$

    This contradicts with $\mathrm{cost}(T'\setminus\{i\})>\mathrm{cost}(T')+B$.
    \end{proof}

    \begin{claim} \label{claim:bound-dv-cost}
        $\mathrm{cost}(S_{D_0,V_0})\le \left(\frac{\xi+\delta}{\xi-\delta}\right)^p\cdot \mathrm{opt}_{\cI}$.
    \end{claim}

    \begin{proof}[Proof of Claim]
        Note that the only difference between $\mathrm{OPT}_{\cI}$ and $S_{D_0,V_0}$ is that each $f^*_i$ in $\mathrm{OPT}_{\cI}$ is replaced by $f_i$ (defined in Algorithm~\ref{alg:solve-sparse}) in $S_{D_0,V_0}$.

        We divide all clients into $|D_0|+1$ groups:

        \begin{itemize}
            \item For each $i\in D_0$, consider all clients in $\mathrm{CBall}_{\cI}(i,\xi L_i)$.

            For such a client $j$, $d(j,f^*_i)$ is upper bounded by $(\xi+\delta)L_i$ while $d(j,f^*_{i'})$ is lower bounded by:

            $$d(j,f^*_{i'})\ge d(i,i')-d(i,j)-d(i',f^*_{i'})\ge (1-\xi-\delta)L_i.$$

            Note that $\xi=1/3$ and $\delta<1/6$, this means $j$ must not connect to $f^*_{i'}$ in the optimal solution, i.e. $j$ connects either to $f^*_i$ or some facility in $V_0$.

            By the definition of $f_i$ in Algorithm~\ref{alg:solve-sparse}, we know that

            $$\sum_{j\in\mathrm{CBall}_{\cI}(i,\xi L_i)}\min(d^p(j,f_i),d^p(j,V_0))\le \sum_{j\in\mathrm{CBall}_{\cI}(i,\xi L_i)}\min(d^p(j,f^*_i),d^p(j,V_0)).$$

            Thus the sum of the connection cost in $\mathrm{CBall}_{\cI}(i,\xi L_i)$ is not larger in $S_{D_0,V_0}$ than in $\mathrm{OPT}_{\cI}$.

            \item Consider all clients in $C\setminus\bigcup_{i\in D_0}\mathrm{CBall}_{\cI}(i,\xi L_i)$.

            For such a client $j$, if $j$ is connected to a facility in $V_0$ in the optimal solution, its connection cost does not increase in $S_{D_0,V_0}$.

            If $j$ is connected to $f^*_i$ in the optimal solution, we know that

            $$\frac{d(j,f_i)}{d(j,f^*_i)}\le\frac{d(j,i)+d(i,f_i)}{d(j,i)-d(i,f^*_i)}\le \frac{\xi+\delta}{\xi-\delta}.$$
        \end{itemize}

        Sum over all clients, we obtain Claim~\ref{claim:bound-dv-cost}.
    \end{proof}

    Combining the above two claims, we obtain Lemma~\ref{lem:solve-sparse}.
\end{proof}

\subsection{Putting everything together}

Now we restate Theorem~\ref{thm:pseudo-to-sol} as follows.

\begin{corollary} \label{cor:pseudo-to-sol}
    For a $k$-clustering instance $\cI$, given a $c$-additive pseudo solution $T$ satisfying
    
    $$\mathrm{cost}(T)\le\alpha\mathrm{opt}_{\cI},$$
    
    there exists an algorithm that runs in time $n^{O((\gamma p)^p\cdot \alpha c/\varepsilon)}$ and returns a solution $S$ such that

    $$\mathrm{cost}(S)\leq(\alpha+\varepsilon)\mathrm{opt}_{\cI},$$

    where $\gamma$ is a global constant.
\end{corollary}

\begin{proof}
    Select the largest $\delta$ such that $\left(\frac{\xi+\delta}{\xi-\delta}\right)^p\leq\alpha$.

    Let $t=\left\lceil\left(\frac{2}{\delta\xi}\right)^p\cdot\frac{4\alpha c}{\varepsilon}\right\rceil$, $A=\frac{\mathrm{opt}_{\cI}}{t}$ and $B=2\cdot\left(A+\frac{\mathrm{cost}(T)}{t}\right)\cdot\left(\frac{2}{\delta\xi}\right)^p$.

    We first invoke Algorithm~\ref{alg:reduce-to-sparse} to obtain an $A$-sparse instance $\cI'$, then invoke Algorithm~\ref{alg:solve-sparse} to obtain the solution $S$. We know that

    $$\begin{aligned}
        \mathrm{cost}(T)+cB&\le \alpha\mathrm{opt}_{\cI}+2c\cdot\left(\frac{\mathrm{opt}_{\cI}}{t}+\frac{\mathrm{cost}(T)}{t}\right)\cdot\left(\frac{2}{\delta\xi}\right)^p\\
        &\le \alpha\mathrm{opt}_{\cI}+4\alpha c\cdot\left(\frac{\mathrm{opt}_{\cI}}{t}\right)\cdot\left(\frac{2}{\delta\xi}\right)^p\\
        &\le (\alpha+\varepsilon)\mathrm{opt}_{\cI}.
    \end{aligned}$$

    Thus from Lemma~\ref{lem:solve-sparse}, we know that $\mathrm{cost}(S)\leq(\alpha+\varepsilon)\mathrm{opt}_{\cI}$.

    From Lemma~\ref{lem:reduce-to-sparse} and Lemma~\ref{lem:solve-sparse}, the running time of the algorithm is $n^{O(t)}$ times that of the pseudo-solution solver. Note that $\xi=1/3$ is a constant, $\delta=\Theta(1/p)$ by definition, then

    $$n^{O(t)}=n^{O((\gamma p)^p\cdot \alpha c/\varepsilon)},$$

    where $\gamma$ is a global constant.
\end{proof}